\documentclass[review,onefignum,onetabnum]{siamonline171218}

\usepackage{lipsum}
\usepackage{amsfonts,amssymb,mathtools,amsmath}
\usepackage{graphicx,subfig}
\usepackage{epstopdf}
\usepackage{onimage}
\usepackage{mathtools,tikz}
\usepackage{algorithmic}
\ifpdf
  \DeclareGraphicsExtensions{.eps,.pdf,.png,.jpg}
\else
  \DeclareGraphicsExtensions{.eps}
\fi

\usepackage{enumitem}
\setlist[enumerate]{leftmargin=.5in}
\setlist[itemize]{leftmargin=.5in}

\newcommand{\Tcal}{\mathcal{T}}
\newcommand{\Mcal}{\mathcal{M}}
\newcommand{\Gcal}{\mathcal{G}}

\newcommand{\vx}{\mathbf{x}}
\newcommand{\ve}{\mathbf{e}}
\newcommand{\vf}{\mathbf{f}}

\newcommand{\vh}{\mathbf{h}}
\newcommand{\vq}{\mathbf{q}}
\newcommand{\vu}{\mathbf{u}}
\newcommand{\vy}{\mathbf{y}}
\newcommand{\vz}{\mathbf{z}}
\newcommand{\vlam}{\pmb{\lambda}}
\newcommand{\vxi}{\pmb{\xi}}

\newcommand{\mI}{\mathbf{I}}
\newcommand{\mC}{\mathbf{C}}
\newcommand{\mB}{\mathbf{B}}
\newcommand{\mA}{\mathbf{A}}
\newcommand{\mH}{\mathbf{H}}
\newcommand{\mR}{\mathbf{R}}
\newcommand{\mS}{\mathbf{S}}

\newcommand{\mL}{\mathbf{L}}

\newcommand{\mPhi}{\pmb{\Phi}}
\newcommand{\mPsi}{\pmb{\Psi}}

\newcommand{\minv}{\left(\mPsi^\intercal \mPhi\right)^{-1}}
\newcommand{\minvT}{\left(\mPhi^\intercal \mPsi\right)^{-1}}
\newsiamremark{remark}{Remark}
\newsiamremark{hypothesis}{Hypothesis}
\crefname{hypothesis}{Hypothesis}{Hypotheses}
\newsiamthm{claim}{Claim}

\headers{Non-intrusive optimization of reduced-order models}{A. Padovan et al.}

\title{Data-driven model reduction via non-intrusive optimization of projection operators and reduced-order dynamics}

\author{Alberto Padovan\thanks{Department of Aerospace Engineering, University of Illinois at Urbana-Champaign, Urbana, IL, 61801 (\email{padovan3@illinois.edu})}
\and Blaine Vollmer\footnotemark[1]\and Daniel J. Bodony\footnotemark[1]}

\usepackage{amsopn}

\makeatletter
\newcommand*{\addFileDependency}[1]{% argument=file name and extension
  \typeout{(#1)}% latexmk will find this if $recorder=0 (however, in that case, it will ignore #1 if it is a .aux or .pdf file etc and it exists! if it doesn't exist, it will appear in the list of dependents regardless)
  \@addtofilelist{#1}% if you want it to appear in \listfiles, not really necessary and latexmk doesn't use this
  \IfFileExists{#1}{}{\typeout{No file #1.}}% latexmk will find this message if #1 doesn't exist (yet)
}
\makeatother

\ifpdf
\hypersetup{
  pdftitle={An Example Article},
  pdfauthor={D. Doe, P. T. Frank, and J. E. Smith}
}
\fi

\begin{document}
\nolinenumbers

\maketitle

% REQUIRED
\begin{abstract} 
Computing reduced-order models using non-intrusive methods is particularly attractive for systems that are simulated using black-box solvers. 
However, obtaining accurate data-driven models can be challenging, especially if the underlying systems exhibit large-amplitude transient growth. 
Although these systems may evolve near a low-dimensional subspace that can be easily identified using standard techniques such as Proper Orthogonal Decomposition (POD), computing accurate models often requires projecting the state onto this subspace via a non-orthogonal projection.
While appropriate oblique projection operators can be computed using intrusive techniques that leverage the form of the underlying governing equations, purely data-driven methods currently tend to achieve dimensionality reduction via orthogonal projections, and this can lead to models with poor predictive accuracy.
In this paper, we address this issue by introducing a non-intrusive framework designed to simultaneously identify oblique projection operators and reduced-order dynamics. 
In particular, given training trajectories and assuming reduced-order dynamics of polynomial form, we fit a reduced-order model by solving an optimization problem over the product manifold of a Grassmann manifold, a Stiefel manifold, and several linear spaces (as many as the tensors that define the low-order dynamics). 
Furthermore, we show that the gradient of the cost function with respect to the optimization parameters can be conveniently written in closed-form, so that there is no need for automatic differentiation. 
We compare our formulation with state-of-the-art methods on three examples: a three-dimensional system of ordinary differential equations, the complex Ginzburg-Landau (CGL) equation, and a two-dimensional lid-driven cavity flow at Reynolds number $Re = 8300$.
% We compare our framework with existing non-intrusive (or weakly intrusive) model reduction techniques such as Operator Inference and POD-Galerkin, and also with the intrusive Trajectory-based Optimization of Oblique Projections (TrOOP) approach.
% On all three examples, we compare the predictive accuracy of our models against state-of-the-art non-intrusive (or weakly intrusive) formulations such as Operator Inference and POD-Galerkin. 
% On the toy model and the CGL equation, we also compare against the highly intrusive, but very accurate, Trajectory-based Optimization of Oblique Projections (TrOOP) framework.
\end{abstract}

% REQUIRED

% REQUIRED

\section{Introduction}

Computing reduced-order models of high-dimensional systems is often necessary to perform several tasks, including accelerating expensive simulations, developing control strategies and solving design optimization problems. 
Most model reduction frameworks share the following key ingredients: a possibly nonlinear map from the high-dimensional state space to a low-dimensional space (i.e., an encoder), a possibly nonlinear map from the low-dimensional space to the original high-dimensional space (i.e., a decoder), and reduced-order dynamics to evolve the reduced-order state. 
Here, we provide a brief review of intrusive and non-intrusive methods where the reduced-order dynamics are continuous in time, and where the encoder and decoder define linear projection operators (i.e., the encoder and decoder are linear maps and the encoder is a left-inverse of the decoder).

Perhaps the most well-known reduced-order models that fall within this category are the so-called linear-projection Petrov-Galerkin models.
These are obtained by (obliquely) projecting the full-order dynamics onto a low-dimensional linear subspace.
In particular, given a decoder $\mPhi \minv$ and an encoder $\mPsi^\intercal$, where $\mPhi$ and $\mPsi$ are tall rectangular matrices that define a projection $\mathbb{P} = \mPhi\minv \mPsi^\intercal$, the aforementioned linear subspace is given by the span of $\mPhi$, and $\mPsi$ specifies the direction of projection.
% maps a high-dimensional vector onto the span of $\mPhi$ along the direction given by the orthogonal complement of the span of $\mPsi$. 
This is neatly illustrated in figure 1 in \cite{Rowley2017model}. 
% In particular, this subspace is given by the span of the decoder $\mPhi\minv$, and the encoder 
% and they are uniquely defined by the orthogonal complement of the nullspace of the encoder $\mPsi^\intercal$ and by the span of the decoder $\mPhi\minv$, where $\mPhi$ and $\mPsi$ are tall rectangular matrices that define an oblique projection operator $\mathbb{P} = \mPhi\minv \mPsi^\intercal$.
If $\mPhi = \mPsi$, then the projection $\mathbb{P}$ is orthogonal and the model is known as a Galerkin model. 
In the simplest of cases, a Galerkin model can be obtained by orthogonally projecting the dynamics onto the span of the leading Proper Orthogonal Decomposition (POD) modes of a representative training data set. 
This procedure is ``weakly'' intrusive in the sense that it requires access to the governing equations, but not necessarily to the linearization and adjoint of the underlying nonlinear dynamics.
In the context of fluids, POD-Galerkin models have been used extensively for both compressible and incompressible flows \cite{Rowley2004model,Noack2003hierarchy,barone2009,Rowley2017model}.
However, these models may not perform well in systems that exhibit large-amplitude transient growth. 
Examples of such systems in fluid mechanics include boundary layers, mixing layers, jets and high-shear flows in general \cite{chomaz2005}. 
The difficulty posed by these systems can often be traced back to the non-normal nature of the underlying linear dynamics, which demands the use of carefully chosen oblique projections. 
In linear systems, or nonlinear systems that evolve near a steady state, this problem can be addressed using methods such as Balanced Truncation \cite{Moore1981principal,Dullerud2000robust,willcox2002} or Balanced POD \cite{Rowley2005model}, which produce oblique projection operators and corresponding Petrov-Galerkin models by balancing the observability and reachability Gramians associated with the underlying linear dynamics. 
Extensions and variants of Balanced Truncation and Balanced POD also exist for quadratic-bilinear systems \cite{Benner2017balanced} and for systems that evolve in the proximity of time-periodic orbits \cite{varga2000,Ma2010,padovan2023}.
For highly nonlinear systems that lie outside the region of applicability of Balanced Truncation or Balanced POD, one can turn to recently-developed methods such as Trajectory-based Optimization of Oblique Projections (TrOOP) \cite{Otto2022optimizing} and  Covariance
Balancing Reduction using Adjoint Snapshots (COBRAS) \cite{Otto2023cobras}.
TrOOP identifies optimal oblique projections for Petrov-Galerkin modelling by training against trajectories generated by the full-order model, while COBRAS identifies oblique projections for model reduction by balancing the state and gradient covariances associated with the full-order solution map. 
We shall see that our non-intrusive formulation is closely related to TrOOP, so we will discuss the latter in more detail in section \ref{subsec:connection}.
All these Petrov-Galerkin methods are intrusive:
not only do they require access to the full-order dynamics, but also to their linearization about steady or time-varying base flows and to the adjoint of the linearized dynamics. 
Thus, they are not easily applicable to systems that are simulated using black-box solvers. 

Among existing techinques to obtain data-driven reduced-order models with continuous-time dynamics on linear subspaces, the most well-known is perhaps Operator Inference \cite{Peherstorfer2016,kramer2024}. 
Operator Inference fits a model to data by minimizing the difference between (usually polynomial) reduced-order dynamics and the projection of the time-derivative of the full-order state onto a low-dimensional subspace. 
Usually, this subspace is defined by the span of POD modes, and the high-dimensional data are projected orthogonally onto it.
While Operator Inference has been shown to work well for systems that evolve in close proximity of an attractor (see, e.g., \cite{quian2022}), it may suffer from the aforementioned drawbacks of orthogonal projections when applied to highly non-normal systems evolving far away from an attractor (e.g., during transients).
This will become apparent in the examples sections.
There exist several other non-intrusive model reduction frameworks in the literature (e.g., discrete-time formulations, autoencoders parameterized via neural networks, and many others), and we will mention those that are more closely connected with our formulation as needed throughout the manuscript.

In this paper, we introduce a novel non-intrusive framework to address the problems associated with orthogonal projections.
In particular, given training trajectories from the full-order model, we fit an optimal low-order model by simultaneously seeking reduced-order dynamics $\vf_r$ and oblique projection operators $\mathbb{P}$ defined by a linear encoder $\mPsi^\intercal$ and a linear decoder $\mPhi\minv$.
We shall see that the optimization parameters are the subspace $V = \text{Range}(\mPhi)$, which lives naturally on the Grassmann manifold, the matrix $\mPsi$, which can be taken to be an element of the orthogonal Stiefel manifold, and the parameters that define the reduced-order dynamics (e.g., reduced-order tensors if the dynamics are taken to be polynomial).
% We shall see that $\mPsi$ can be taken to be an element of the (orthogonal) Stiefel manifold, and that our reduced-order models are functions of the subspace $V = \text{Range}(\mPhi)$, which is an element of the Grassmann manifold.
Furthermore, if we constrain the reduced-order dynamics $\vf_r$ to be of a form  that lends itself to straightforward differentiation (e.g., polynomial), we show that the gradient of the cost function with respect to the optimization parameters can be written in closed-form. 
This is quite convenient because it bypasses the need for automatic differentiation and it allows for faster training. 
We test our formulation on three different examples: a simple system governed by three ordinary differential equations, the complex Ginzburg-Landau (CGL) equation and the two-dimensional incompressible lid-driven cavity flow at Reynolds number $Re = 8300$.
On all three examples, we compare our framework with Operator Inference and POD-Galerkin. 
In the first two examples, we also compare with TrOOP, which has been shown to give very accurate Petrov-Galerkin models in several examples, including highly non-normal and nonlinear jets \cite{Otto2022optimizing,otto2023}.
On all three examples, our models exhibit better performance than Operator Inference and POD-Galerkin models, and in the first two examples we obtain models with predictive accuracy very close to that of the intrusive TrOOP formulation.

\section{Mathematical formulation}
\label{sec:math_form}
Throughout this section, we consider a general nonlinear system with dynamics defined by
\begin{equation}
\label{eq:fom}
    \begin{aligned}
        \frac{\mathrm{d}\vx}{\mathrm{d} t} &= \vf(\vx,\vu),\quad \vx(0) = \vx_0 \\
        \vy &= \vh(\vx)
    \end{aligned}
\end{equation}
where~$\vx \in \mathbb{R}^n$ is the state vector, $\vx_0$ is the initial condition, $\vu \in\mathbb{R}^m$ is the control input and $\vy \in \mathbb{R}^p$ is the measured output.
Since our model reduction procedure draws inspiration from the form of Petrov-Galerkin reduced-order models, we begin by providing a brief review of the latter.
We then introduce our framework in section \ref{subsec:non_intrusive_models}.
% Before introducing our model reduction procedure, it is instructive to provide a brief review of Petrov-Galerkin (or projection-based) models. 

\subsection{Petrov-Galerkin models}
\label{subsec:petrov_galerkin_models}

As discussed in the introduction, Petrov-Galerkin reduced-order models are a class of models obtained by constraining the full-order dynamics in \eqref{eq:fom} to a linear subspace of $\mathbb{R}^n$. 
While Petrov-Galerkin models can also be obtained via nonlinear projection onto curved manifolds \cite{otto2023}, here we constrain our attention to the more common case of linear projections. 
% \bvremark{does this linear projection restriction impact any of our developments later? i.e. would using nonlinear potentially negate any benefits of the new method?}
Given rank-$r$ matrices $\mPhi \in \mathbb{R}^{n\times r}$ and $\mPsi\in\mathbb{R}^{n\times r}$ that define an \emph{oblique} projection $\mathbb{P} = \mPhi\minv \mPsi^\intercal$, the corresponding Petrov-Galerkin model for~\eqref{eq:fom} is given by 
\begin{equation}
    \label{eq:rom}
    \begin{aligned}
        \frac{\mathrm{d} \hat{\vx}}{\mathrm{d} t} &= \mathbb{P}\vf\left(\mathbb{P}\hat{\vx},\vu\right),\quad \hat{\vx}(0) = \mathbb{P}\vx_0 \\
        \hat{\vy} &= \vh\left(\mathbb{P}\hat{\vx}\right),
    \end{aligned}
\end{equation}
where $\hat{\vx}$ lies in the range of $\mathbb{P}$ for all times. 
In the special case of $\mPsi = \mPhi$, the projection $\mathbb{P}$ is orthogonal and the model \eqref{eq:rom} is referred to as a Galerkin model.
While the state $\hat{\vx}\in \mathbb{R}^n$ is an $n$-dimensional vector (i.e., the same size of the original state $\vx$), the dynamics \eqref{eq:rom} can be realized by an equivalent $r$-dimensional system
\begin{equation}
\label{eq:rom_z}
    \begin{aligned}
        \frac{\mathrm{d} \hat{\vz}}{\mathrm{d} t} &= \mPsi^\intercal \vf\left(\mPhi\minv \hat{\vz},\vu\right),\quad \hat{\vz}(0) = \mPsi^\intercal \vx_0 \\
        \hat{\vy} &= \vh\left(\mPhi\minv\hat{\vz}\right)
    \end{aligned}
\end{equation}
where the state vector $\hat{\vz}= \mPsi^\intercal \hat{\vx}$ has dimension $r$. 
The primary challenge associated with computing accurate projection-based reduced-order models lies in identifying matrices~$\mPhi$ and~$\mPsi$ that define appropriate projections $\mathbb{P}$. 
While there exist several methods to address this challenge, these are often intrusive in the sense that they require access to the linearization of \eqref{eq:fom} and its adjoint \cite{Rowley2005model,Otto2022optimizing,Otto2023cobras}. 
In the next section, we present a non-intrusive model reduction formulation by allowing for the reduced-order dynamics to be independent of the full-order right-hand side $\vf$.

\subsection{Non-intrusive optimization of projection operators and reduced-order dynamics}
\label{subsec:non_intrusive_models}

Here, we consider reduced-order models of the form 
\begin{equation}
\label{eq:rom_ni}
    G(\mPhi,\mPsi,\hat\vf_r) = 
    \begin{dcases}
        \begin{aligned}
            \frac{\mathrm{d} \hat{\vz}}{\mathrm{d} t} &=  \vf_r\left(\hat{\vz},\vu\right),\quad \hat{\vz}(0) = \mPsi^\intercal \vx_0 \\
            \hat{\vy} &= \vh\left(\mPhi\minv\hat{\vz}\right)
        \end{aligned}
    \end{dcases}
\end{equation}
It is instructive to observe that if $\vf_r(\hat{\vz},\vu) =\mPsi^\intercal \vf\left(\mPhi\minv \hat{\vz},\vu\right)$ then \eqref{eq:rom_ni} is the exact analog of the Petrov-Galerkin reduced-order model in \eqref{eq:rom_z}.
Instead, we let $\vf_r$ be a general function of the reduced-order state $\hat{\vz}$ and of the input $\vu$.
% While this implies that \eqref{eq:rom_ni} is no longer a projection-based model, we shall see shortly that this is what allows us to make the framework non-intrusive.
So, while Petrov-Galerkin models are fully defined by (the span of) the matrices $\mPhi$ and $\mPsi$ that define a projection onto a low-dimensional subspace, here we have additional degrees of freedom in the choice of the reduced-order dynamics. 
We shall see momentarily that this additional freedom allows us to proceed non-intrusively.

% We now work towards writing an optimization problem that will deliver an optimal reduced-order model of the form of \eqref{eq:rom_ni}. 
Within our framework, we seek reduced-order models of the form of \eqref{eq:rom_ni} by minimizing the error between ground-truth observations $\vy$ coming from \eqref{eq:fom} and the predicted observations $\hat\vy$ given by \eqref{eq:rom_ni}. 
In order to convert this task into an appropriate optimization problem, it is useful to first identify the symmetries and constraints that are present in~\eqref{eq:rom_ni}.
We begin by observing that the system~$G$ in~\eqref{eq:rom_ni} is invariant with respect to a rotation and scaling of the basis matrix $\mPhi$. 
In fact, $G(\mPhi \mR,\mPsi,\vf_r) = G(\mPhi,\mPsi,\vf_r)$ for any invertible matrix~$\mR$ of size $r\times r$.
It follows that the reduced-order system defined by~\eqref{eq:rom_ni} is a function of the $r$-dimensional subspace $V = \text{Range}(\mPhi)$, rather than of the matrix representative $\mPhi$ itself. 
In the mathematical statement of the problem we will make use of this symmetry and leverage the fact that $r$-dimensional subspaces of $\mathbb{R}^n$ are elements of the Grassmann manifold $\mathcal{G}_{n,r}$.
An analogous type of symmetry does \emph{not} hold for $\mPsi$. 
In fact, it can be easily verified that there exist invertible matrices $\mS$ such that $G(\mPhi,\mPsi\mS,\vf_r) \neq G(\mPhi,\mPsi,\vf_r)$. 
While \eqref{eq:rom_ni} does not enjoy any $\mPsi$-symmetries, we still require $\mPsi$ to have full column rank (otherwise the product $\mPsi^\intercal\mPhi$ would be rank deficient).
It is therefore natural to constrain $\mPsi$ to the Stiefel manifold $S_{n,r}$ of orthonormal (and, hence, full-rank) $n\times r$ matrices.
% Finally, the parameters $\mF$ that define the action of $\vf_r$ on the state $\hat\vz$ and input $\vu$ live naturally on a linear manifold $\Ecal$.
% For instance, if $\vf_r(\hat\vz,\vu) = \mA_r \hat\vz$, then $\mF = \text{Vec}(\mA_r)$ (where $\text{Vec}$ denotes the vectorization of a tensor) and $\Ecal = \mathbb{R}^{r^2}$.
Finally, in order to write an optimization problem where the gradient of the cost function with respect to all the parameters can be obtained in closed-form, it is convenient to constrain the reduced-order dynamics $\vf_r$ to a form that lends itself to straightforward differentiation. 
Throughout this paper, we will let~$\vf_r$ be a polynomial function of the reduced-order state $\hat{\vz}$ and of the input $\vu$ as follows
\begin{equation}
\label{eq:fr}
    \vf_r = \underbrace{\mA_r \hat\vz + \mB_r\vu +\mH_r:\hat\vz\hat\vz^\intercal}_{\coloneqq \overline{\vf}_r}  + \mL_r : \hat\vz \vu^\intercal + \ldots.
\end{equation}
Here, capital letters denote reduced-order tensors that lie naturally on linear manifolds of appropriate dimension (e.g., $\mA_r \in \mathbb{R}^{r\times r}$, $\mB_r \in \mathbb{R}^{r\times m}$ and $\mH\in\mathbb{R}^{r\times r\times r}$).
In the interest of a more concise description of the mathematical formulation, we take~$\vf_r = \overline\vf_r$.
This assumption can be trivially relaxed.

We are now ready to state the optimization problem that will give us an optimal reduced-order model of the form of \eqref{eq:rom_ni}.
Given outputs $\vy(t_i)$ sampled at times $t_i$ along a trajectory generated from the full-order system \eqref{eq:fom}, we seek a solution to 
\begin{equation}
\label{eq:opt_problem}
    \begin{aligned}
        \min_{\left(V,\mPsi,\mA_r,\mH_r,\mB_r\right) \in \mathcal{M}}\,\,\, J &=  \sum_{i=0}^{N-1}\lVert \vy(t_i) - \hat{\vy}(t_i)\rVert^2 \\
        \text{subject to: } \quad \frac{d \hat{\vz}}{d t} &=  \overline\vf_r(\hat\vz,\vu),\quad \hat{\vz}(t_0) = \mPsi^\intercal \vx(t_0) \\
            \hat{\vy} &= \vh\left(\mPhi\minv\hat{\vz}\right) \\
            V &= \text{Range}\left(\mPhi\right)
    \end{aligned}
\end{equation}
where $\mathcal{M} = \mathcal{G}_{n,r}\times S_{n,r}\times \mathbb{R}^{r\times r}\times \mathbb{R}^{r\times r\times r} \times \mathbb{R}^{r\times m}$ is the product manifold that defines our optimization domain. 
% \bvremark{where do the symmetry and other restrictions on $\Psi$ and $\Phi$ come in? I probably just missed the point?}

% For the sake of completeness, we observe that higher-order reduced-order dynamics can be obtained by appending the appropriate higher-order tensors to the parameters tuple. 

\subsection{Gradient-based optimization on $\Mcal$}

In order to solve the optimization problem \eqref{eq:opt_problem} using a gradient-based algorithm, it is necessary to endow the manifold $\mathcal{M}$ with a Riemannian metric $g^\mathcal{M}$. 
This is a smooth family of inner products $g^{\mathcal{M}}_p$ defined on the tangent spaces of the manifold $\mathcal{M}$,
\begin{equation}
    g^\mathcal{M}_p : \mathcal{T}_p\mathcal{M}\times \mathcal{T}_{p}\mathcal{M} \to \mathbb{R},
\end{equation}
where $\mathcal{T}_p\mathcal{M}$ denotes the tangent space of $\mathcal{M}$ at a point $p\in\Mcal$ \cite{Absil2009optimization}. 
The gradient $\xi$ of the cost function at $p\in \mathcal{M}$ is then defined as the element of the tangent space $\Tcal_p \Mcal$ that satisfies 
\begin{equation}
    D J[\eta] = g_p^{\Mcal}\left(\xi,\eta\right),\quad \forall \eta \in \Tcal_p^\Mcal,
\end{equation}
where $D J[\eta]$ is the directional derivative.
To define a metric for the product manifold $\mathcal{M}$, we first need to define metrics for each individual component of the Cartesian product that defines $\mathcal{M}$. 

We begin with the Stiefel manifold. 
Since $S_{n,r}$ can be viewed as an embedded submanifold of $\mathbb{R}^{n\times r}$, it is natural to endow $S_{n,r}$ with the metric
\begin{equation}
    g^{S_{n,r}}_{\mPsi}\left(\xi,\eta\right) = \text{Tr}\left(\xi^\intercal \eta\right),\quad \xi,\,\eta \in \mathcal{T}_{\mPsi}S_{n,r},
\end{equation}
which is the Euclidean metric inherited from the ambient space $\mathbb{R}^{n\times r}$ \cite{Absil2009optimization} and $\text{Tr}$ denotes the trace. 
A metric for the Grassmann manifold can be defined analogously, albeit paying attention to the fact that the Grassmannian is an abstract manifold with non-unique matrix representatives.
If we view $\Gcal_{n,r}$ as a quotient manifold of the non-orthogonal Stiefel manifold $\mathbb{R}^{n\times r}_*$ (which is the manifold of rank-$r$, but not necessarily orthonormal, matrices of size $n\times r$), we can define a metric on the abstract manifold $\Gcal_{n,r}$ in terms of matrix-valued objects defined in the ambient space $\mathbb{R}^{n\times r}_{*}$. 
This is convenient since elements of the ambient space can be easily represented on a computer. 
Thus, given the ambient space metric 
\begin{equation}
\label{eq:Phi_metric}
    g^{\mathbb{R}_{*}^{n\times r}}_{\mPhi}\left(\xi,\eta\right) = \text{Tr}\left(\left(\mPhi^\intercal \mPhi\right)^{-1}\xi^\intercal \eta \right),\quad \xi,\, \eta \in \Tcal_{\mPhi} \mathbb{R}^{n\times r}_{*}
\end{equation}
we let the metric on $\Gcal_{n,r}$ be defined as
\begin{equation}
\label{eq:metric_grass}
    g^{\Gcal_{n,r}}_{V}(\xi,\eta) = g^{\mathbb{R}_{*}^{n\times r}}_{\mPhi}\left(\overline{\xi}_{\mPhi},\overline{\eta}_{\mPhi}\right),\quad \xi,\,\eta \in \Tcal_V,\quad \text{Range}\left(\mPhi\right) = V.
\end{equation}
It is worth observing that \eqref{eq:metric_grass} is not yet suited for computation, since there exists an infinite number of elements $\overline{\xi}_{\mPhi}$ and $\overline{\eta}_{\mPhi}$ of $\Tcal_{\mPhi} \mathbb{R}^{n\times r}_{*}$ that satisfy the equality.
The ambiguity is removed by requiring $\overline{\xi}_{\mPhi}$ and $\overline{\eta}_{\mPhi}$ to lie on the \emph{horizontal space}, which is a subspace of $\Tcal_{\mPhi} \mathbb{R}^{n\times r}_{*}$ where one may identify unique $\overline{\xi}_{\mPhi}$ and $\overline{\eta}_{\mPhi}$ that satisfy \eqref{eq:metric_grass}.
A rigorous characterization of the horizontal space is provided in chapter 3 of \cite{Absil2009optimization}, and the specific case of the Grassmann manifold is considered in example 3.6.4 in the same reference. 
Finally, for the linear manifolds in the Cartesian product of $\Mcal$, we adopt the Euclidean metric (i.e., the usual tensor dot product).

Another ingredient that is necessary for gradient-based manifold optimization is the concept of a \emph{retraction}. 
This is a map $R_p: \Tcal_p \Mcal \to \Mcal$ that satisfies $R_p(0) = p$ and $D R_p(0) = I_{\Tcal_p\Mcal}$, where $I_{\Tcal_p\Mcal}$ is the identity map on the tangent space $\Tcal_p\Mcal$ \cite{Absil2009optimization}.
The use of this map allows us to generalize the concept of moving in the direction of the gradient on a nonlinear manifold: for instance, given a point $p \in \Mcal$ and the gradient $\xi \in \Tcal_p\Mcal$ of a function $f$ defined on $\Mcal$, the next iterate in the direction of the gradient is given by $R_p(p - \alpha \xi) \in \Mcal$, where~$\alpha$ is some learning rate.
In other words, the retraction allows us to guarantee that all iterates generated by a gradient flow lie on the manifold. 
Valid retractions for both the Stiefel and Grassmann manifolds are given by the QR decomposition (see examples 4.1.3 and 4.1.5 in \cite{Absil2009optimization}), while for linear manifolds the retraction is simply the identity map. 
Lastly, we point out that second-order gradient-based algorithms (e.g., conjugate gradient) require the concept of \emph{vector transport}.
This is thoroughly described in section 8.1 of \cite{Absil2009optimization}. 

% \bvremark{Consider putting a reference for all the math to help readers unfamiliar with all the jargon (like me) at the beginning of the section?}

Gradient-based algorithms on nonlinear manifolds are well-understood and readily available in libraries such as \texttt{Pymanopt} \cite{pymanopt} in Python or \texttt{Manopt} \cite{manopt2014} in MATLAB.
Metrics, retractions and vector transports are conveniently handled by these packages, and a user simply needs to provide routines to evaluate the cost function and the Euclidean gradient. 
% that only require the definition of the cost function and of the Euclidean gradient from the user. 
Conveniently, our model reduction formulation allows for the Euclidean gradient of the cost function in \eqref{eq:cost_function} to be computed in closed form. 
This result is stated in the following proposition, where we also see that the computation of the gradient does not require querying the full-order model \eqref{eq:fom}. 
In other words, the gradient can be computed non-intrusively. 
% a closed-form equation for the Euclidean gradient of the cost function $J$ in \eqref{eq:opt_problem} with respect to the problem parameters.

\begin{proposition}
\label{prop:gradient}
    Let \eqref{eq:opt_problem} be written as an equivalent unconstrained optimization problem where we seek a minimum to the Lagrangian
    \begin{equation}
    \label{eq:lagrangian}
        \begin{aligned}
            L =  \sum_{i=0}^{N-1}&\bigg\{\big\lVert \vy(t_i) - \vh\left(\mPhi\minv\hat{\vz}(t_i)\right)\big\rVert^2  + \int_{t_0}^{t_i} \vlam_i^\intercal\left(\frac{d \hat{\vz}}{d t} -  \mA_r \hat\vz -\mH_r:\hat\vz\hat\vz^\intercal - \mB_r\vu\right)d t\bigg. \\
            &\bigg. + \vlam_i(t_0)^\intercal \left(\hat\vz(t_0) - \mPsi^\intercal \vx(t_0)\right)\bigg\}
        \end{aligned}
    \end{equation}
    Here, $\vlam_i(t) \in\mathbb{R}^r$ with $t\in [t_0,t_i]$ is the $i$th Lagrange multiplier.
    Defining $\ve(t_i) \coloneqq \vy(t_i) - \vh\left(\mPhi\minv\hat{\vz}(t_i)\right)$ and $\mC_{j,k} \coloneqq \partial \vh_j/\partial \vx_k$, the Euclidean gradients of the Lagrangian with respect to each of the parameters are given below,
    \begin{align}
        \nabla_{\mPhi} L &= -2 \sum_{i=0}^{N-1}\left(\mI 
 - \mPsi \minvT \mPhi^\intercal\right)\mC(t_i)^\intercal \ve(t_i)\hat{\vz}(t_i)^\intercal \minv \label{eq:grad_Phi}\\
        \nabla_{\mPsi} L &= \sum_{i=0}^{N-1}\left(2 \mPhi \minv \hat\vz(t_i)\ve(t_i)^\intercal \mC(t_i)^\intercal \mPhi \minv - \vx(t_0)\vlam_i(t_0)^\intercal\right) \\
        \nabla_{\mA_r} L &= -\sum_{i=0}^{N-1}\int_{t_0}^{t_i}\vlam_i \hat\vz ^\intercal\,dt \\
        \nabla_{\mH_r} L &= -\sum_{i=0}^{N-1}\int_{t_0}^{t_i}\vlam_i \otimes \left(\hat\vz \hat\vz^\intercal\right)\,dt\\
        \nabla_{\mB_r} L &= -\sum_{i=0}^{N-1}\int_{t_0}^{t_i}\vlam_i \vu^\intercal\,dt,
    \end{align}
    where the Lagrange multiplier $\vlam_i(t)$ satisfies the reduced-order adjoint equation
    \begin{equation}
    \label{eq:adjoint_rom}
        -\frac{d \vlam_i}{dt} = \left[\partial_{\hat\vz}\overline{\vf}_r(\hat\vz)\right]^\intercal \vlam_i,\quad \vlam_i(t_i) = 2 \minvT \mPhi^\intercal \mC(t_i)^\intercal \ve(t_i),\quad t\in [t_0,t_i].
    \end{equation}
\end{proposition}
\begin{proof}
    The proof is given in appendix \ref{app:proof}. In the interest of generality, it is worth mentioning that if the reduced-order dynamics are, e.g., cubic or higher-order, then the gradient of $L$ with respect to the higher-order tensors can be computed in a similar fashion. 
\end{proof}

% The proposition shows that the gradient of the cost function with respect to the optimization parameters does not require querying the full-order dynamics that generated the training trajectory $\vy(t)$, so the optimization problem can be solved non-intrusively. 
% that the optimization problem can indeed be solved non-intrusively since no gradient evaluation requires querying the full-order right-hand side $\vf$ in \eqref{eq:fom} or its linearization.

\subsection{Connection with existing methods}
\label{subsec:connection}

While our model reduction framework shares similarities with several existing methods, we would like to emphasize a natural connection with the recently-developed Trajectory-based Optimization for Oblique Projections (TrOOP) \cite{Otto2022optimizing} and the Operator Inference framework introduced in~\cite{Peherstorfer2016}.

TrOOP is a model reduction framework whereby a Petrov-Galerkin reduced-order model of the form \eqref{eq:rom} is obtained by optimizing the projection operator $\mathbb{P}$ against trajectories of the full-order model \eqref{eq:fom}. 
More specifically, given $r$-dimensional subspaces $V = \text{Range}(\mPhi)$ and $W = \text{Range}(\mPsi)$, TrOOP seeks an optimal $\mathbb{P}$ by solving the following optimization problem 
\begin{equation}
\label{eq:opt_problem_troop}
    \begin{aligned}
        \min_{\left(V,W\right) \in \Mcal_{\text{TrOOP}}}\,\,\, J_{\text{TrOOP}} &=  \sum_{i=0}^{N-1}\lVert \vy(t_i) - \hat{\vy}(t_i)\rVert^2
    \end{aligned}
\end{equation}
subject to \eqref{eq:rom} (or, equivalently, to \eqref{eq:rom_z}), where $\mathcal{M}_{\text{TrOOP}} = \Gcal_{n,r}\times \Gcal_{n,r}$ is the product of two Grassmann manifolds. 
While the cost function \eqref{eq:opt_problem_troop} is the same as the one in \eqref{eq:opt_problem}, solving the optimization problem \eqref{eq:opt_problem_troop} is intrusive because TrOOP constrains the reduced-order dynamics to be the Petrov-Galerkin projection of the full-order dynamics.
Consequently, computing the gradient of the cost function $J_{\text{TrOOP}}$ with respect to the parameters requires differentiating through the dynamics $\vf$ in \eqref{eq:fom}.
This can be seen by deriving the gradient in a way analogous to that of Proposition \ref{prop:gradient}, or alternatively, following Proposition 4.3 in \cite{Otto2022optimizing}.
As previously discussed, not all black box solvers allow for easy differentiation of the governing equations so, for this reason, solving the TrOOP optimization problem can be infeasible in some applications. 

Operator Inference, on the other hand, is a non-intrusive model reduction framework that seeks a reduced-order model by orthogonally projecting the data onto a low-dimensional subspace and then fitting the reduced-order dynamics. 
This subspace is typically chosen as the span of the leading Proper Orthogonal Decomposition (POD) modes associated with some representative data set generated from~\eqref{eq:fom}. 
In particular, given a full-order trajectory $\vx(t_i)$ sampled from~\eqref{eq:fom} at times $t_i$, the time-derivative $\mathrm{d} \vx(t_i)/\mathrm{d}t$, the input $\vu(t_i)$, a $r$-dimensional subspace spanned by $\mPhi \in \mathbb{R}^{n\times r}$, and some parameterization of the reduced-order dynamics (e.g., $\vf_r = \mA_r \hat\vz + \mH_r \left(\hat\vz\otimes\hat\vz\right) + \mB_r\vu$), Operator Inference solves 
\begin{equation}
\label{eq:opt_problem_opinf}
    \min_{\left(\mA_r,\mH_r,\mB_r\right) \in \Mcal_{\text{OpInf}}} \quad J_{\text{OpInf}} = \sum_{i=0}^{N-1} \bigg\lVert \frac{\mathrm{d}\hat\vz(t_i)}{\mathrm{d}t} - \mA_r \hat\vz(t_i) - \mH_r: \hat\vz(t_i)\hat\vz(t_i)^\intercal - \mB_r\vu(t_i)\bigg\rVert^2,
\end{equation}
where $\hat\vz(t_i) = \mPhi^\intercal \vx(t_i)$ and $\Mcal_{\text{OpInf}} = \mathbb{R}^{r\times r}\times \mathbb{R}^{r\times r\times r} \times \mathbb{R}^{r\times m}$.
As observed in \cite{Peherstorfer2016}, equation \eqref{eq:opt_problem_opinf} can be conveniently written as a linear least-squares problem whose solution is obtained via the Moore-Penrose inverse rather than via iterative gradient-based algorithms.
Furthermore, given the least-squares nature of the problem, it is straightforward to add regularization (e.g., to promote stability and/or avoid overfitting) by penalizing the Frobenius norm of the parameters \cite{mcquarrie2021}.
While Operator Inference offers a convenient non-intrusive model reduction platform, it may suffer from the fact that it maps the high-dimensional data onto a low-dimensional space via orthogonal projection. 
We shall see that this can lead to inaccurate models if the full-order dynamics exhibit transient growth (e.g., due to non-normal mechanisms). 
% Furthermore, the Operator Inference procedure may be prone to producing unstable reduced-order models since it fits the reduced-order dynamics against the time-derivative of the reduced-order state, rather than against the reduced-order state itself. 
% This can become problematic in the presence of noise or of extremely fast time scales (e.g., in high-speed flows)
% While this issue can be mitigated by penalizing the Frobenius norm of the parameters via Tikhonov regularization \cite{mcquarrie2021}, it can nonetheless lead to poor models in the presence of noise or of extremely fast time scales. 

It is now clear that our model reduction framework merges concepts from both TrOOP and Operator Inference. 
Specifically: TrOOP seeks optimal projections while constraining the reduced-order dynamics to be of Petrov-Galerkin form, Operator Inference seeks optimal reduced-order dynamics while constraining the projection operator to be orthogonal and onto the span of POD modes, and our formulation simultaneously seeks optimal projections and optimal reduced-order dynamics. 
% \bvremark{Is it possible to put this description of the models into an equation form too?}
We thus call our formulation ``Non-intrusive Trajectory-based optimization of Reduced-Order Models'' (NiTROM). 
In closing this section, it is also worth mentioning that NiTROM solves an optimization problem similar in spirit to the one in ``low-rank dynamic mode decomposition'' \cite{sashittal2019}, where the encoder and decoder are taken to be elements of the Grassmann manifold, and the reduced-order dynamics are assumed to be linear and discrete in time. 
Furthermore, by viewing the projection operator as a linear autoencoder, we can find several connections between NiTROM and existing intrusive and non-intrusive model reduction formulations that rely on (usually nonlinear) autoencoders parameterized by neural networks. 
Recent examples may be found in \cite{fresca2021,conti2023,otto2023}, although, to the best of our understanding, the only autoencoder architecture that defines a nonlinear projection onto a curved manifold is presented in \cite{otto2023}.

\section{Application to a toy model}
\label{sec:toy_model}

In this section, we apply NiTROM to a three-dimensional toy model, and we compare with the intrusive TrOOP and POD Galerkin formulations and the non-intrusive Operator Inference.  
The model is governed by the following equations 
\begin{align}
    \dot{x}_1 &= -x_1 + 20 x_1 x_3 + u \label{eq:x1} \\
    \dot{x}_2 &= -2 x_2 + 20 x_2 x_3 + u \label{eq:x2} \\
    \dot{x}_3 &= -5 x_3 + u \label{eq:x3} \\
    y &= x_1 + x_2 + x_3, \label{eq:y}
\end{align}
where $\dot{x}_1 = \mathrm{d} x_1/\mathrm{d}t$. 
As discussed in \cite{Otto2022optimizing,Otto2023cobras} these dynamics are particularly tedious and exhibit large-amplitude transient growth because the rapidly-decaying state $x_3$ has a large impact on the remaining states through the quadratic nonlinearity. 
Here, we seek a two-dimensional reduced-order model capable of predicting the time-history of the measured output $y$ in response to step inputs $u \in [0.01,\beta)$, where $\beta$ is defined momentarily. 
For a given input $u$, the system admits a steady state solution with $x_3 = u/5$.
Substituting $x_3$ into \eqref{eq:x1} and \eqref{eq:x2}, it is easy to verify that the forced steady state is stable if and only if 
\begin{equation}
    -1 + \frac{20}{5}u < 0 \quad \text{and}\quad -2 + \frac{20}{5}u < 0.
\end{equation}
Thus, the system \eqref{eq:x1}--\eqref{eq:x3} admits bounded step responses so long as $u < 5/20 = 0.25$. 
We therefore take $\beta = 0.25$.
Given the quadratic nature of the full-order dynamics, we seek quadratic reduced-order models of the form 
\begin{align}
    \frac{\mathrm{d}\hat{\vz}}{\mathrm{d}t} &= \mA_r \hat{\vz} + \mH_{r}:\hat\vz\hat\vz^\intercal + \mPsi^\intercal \vu \label{eq:rom_toy_model}\\
    \hat{y} &= \mC \mPhi\minv \hat\vz,
\end{align}
where $\mC = \left[1\,1\,1\right]$ is a row vector and $\vu = \left(u,u,u\right)$. 

\begin{figure}
\centering
\begin{minipage}{0.48\textwidth}
\begin{tikzonimage}[trim= 10 15 10 10,clip,width=0.85\textwidth]{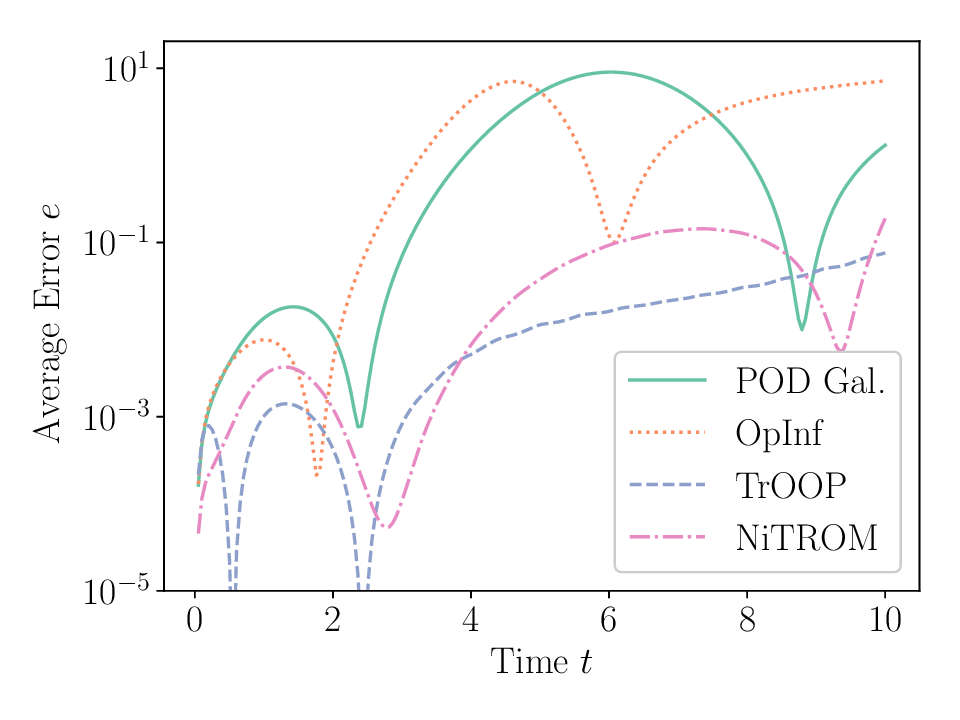}%[tsx/show help lines]
\node at (0.23,0.88) {\small $\textit{(a)}$};
\end{tikzonimage}
\end{minipage}
% \hspace{-4ex}
\begin{minipage}{0.48\textwidth}
\begin{tikzonimage}[trim= 10 15 10 10,clip,width=0.85\textwidth]{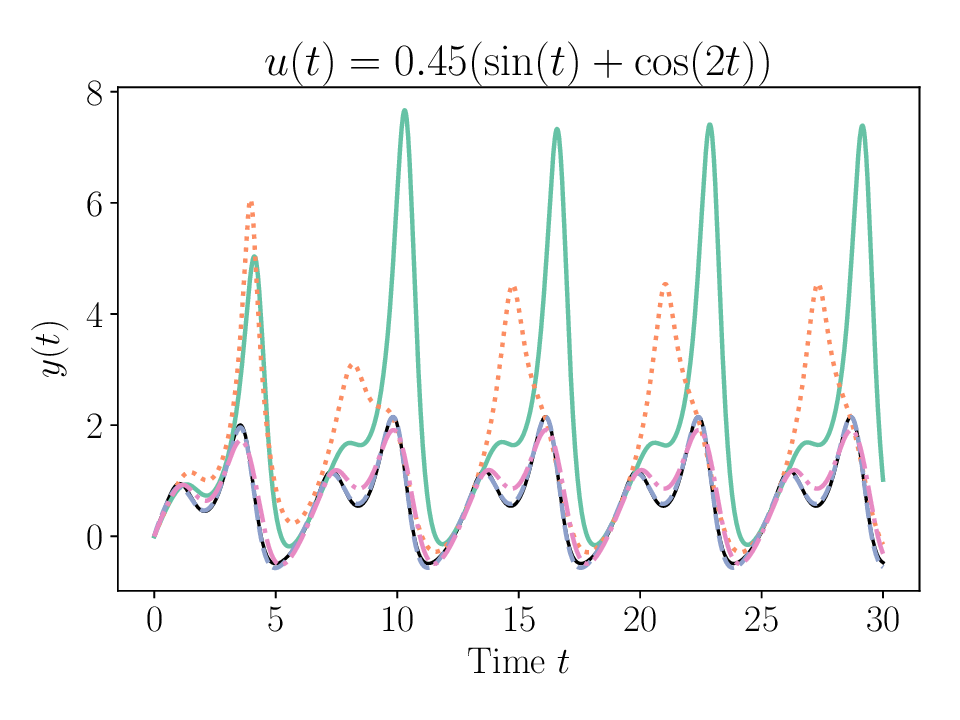}%[tsx/show help lines]
\node at (0.17,0.81) {\small $\textit{(b)}$};
\end{tikzonimage}
\end{minipage}
% \begin{tikzonimage}[trim= 10 80 10 110,clip,width=0.6\textwidth]{Figures_toy_model_new/cost_vals_comparison}%[tsx/show help lines]
% \node at (0.19,0.4) {\small $\textit{(c)}$};
% \end{tikzonimage}
\caption{Panel (a) shows the average testing error defined in \eqref{eq:avg_error}. Panel (b) shows the time history of the output $y$ in response to a sinusoidal input $u$. The black line in panel (b) denotes the ground-truth response.}
\label{fig:toy_model}
\end{figure}

In order to train the models, we collect $y(t)$ from $N_{\text{traj}} = 4$ step responses generated with $u \in \{0.01,0.1,0.2,0.248\}$ and initialized from rest.
For each trajectory, we sample~$y$ at $N = 20$ equally-spaced times $t_i \in [0,10]$. 
The cost function for NiTROM and TrOOP is  
\begin{equation}
\label{eq:cost_function}
    J = \sum_{j=0}^{N_{\text{traj}}-1} \frac{1}{\alpha_j} \sum_{i=0}^{N-1} \lVert y^{(j)}(t_i) - \hat{y}^{(j)}(t_i)\rVert^2,
\end{equation}
with $\alpha_j = N_{\text{traj}} N \lVert \mC\overline{\vx}^{(j)}\rVert^2$, where $\overline{\vx}^{(j)}$ is the exact steady state that arises in response to the step input $u^{(j)}$. 
For both methods, the optimization was performed using the conjugate gradient algorithm available in \texttt{Pymanopt} \cite{pymanopt}, with the Euclidean gradient defined following Proposition \ref{prop:gradient}. 
Both methods were initialized with $\mPsi = \mPhi$ given by the leading two POD modes computed from the four training step responses.
Additionally, NiTROM was provided with initial reduced-order tensors computed via Galerkin projection of the full-order dynamics onto the POD modes.
The cost function for Operator Inference is 
\begin{equation}
    \begin{aligned}
        J_{\text{OpInf}} &= \sum_{j=0}^{N_\text{traj}-1}\frac{1}{\alpha_j}\sum_{i=0}^{N-1} \bigg\lVert \frac{\mathrm{d}\hat\vz^{(j)}(t_i)}{\mathrm{d}t} - \mA_r \hat\vz^{(j)}(t_i) - \mH_r : \hat\vz^{(j)}(t_i)\hat\vz^{(j)}(t_i)^\intercal - \mPhi^\intercal\vu^{(j)}(t_i)\bigg\rVert^2 \\
        &+\lambda \lVert \text{Mat}(\mH_r)\rVert^2_F,
    \end{aligned}
\end{equation}
where $\mPhi$ are the POD modes that we just described, $\hat\vz = \mPhi^\intercal \vx$, $\text{Mat}\left(\mH_r\right)$ denotes the matricization of the third-order tensor $\mH_r$ and $\lambda = 10^{-6}$ is the regularization parameter. 
The chosen $\lambda$ is (approximately) the one that yields the best possible Operator Inference model, as measured by the cost function $J$ in \eqref{eq:cost_function}.
% The parameter $\lambda$ was chosen such that the reduced-order dynamics given by Operator Inference would be stable for all inputs $u \in [0,\beta)$.
% No penalty was required for $\mA_r$.
% Finally, the POD-Galerkin model is obtained via Galerkin projection of the full-order model using the same POD modes $\mPhi$ used for Operator Inference. 

% \begin{figure}
% \centering
% \begin{minipage}{0.48\textwidth}
% \begin{tikzonimage}[trim= 10 15 10 10,clip,width=0.87\textwidth]{Figures_toy_model/toy_model_avg_error_plot}%[tsx/show help lines]
% \node at (0.9,0.20) {\small $\textit{(a)}$};
% \end{tikzonimage}
% \begin{tikzonimage}[trim= 10 15 10 10,clip,width=0.9\textwidth]{Figures_toy_model/toy_model_trajectory_04}%[tsx/show help lines]
% \node at (0.25,0.2) {\small $\textit{(c)}$};
% \end{tikzonimage}
% \end{minipage}
% % \hspace{-4ex}
% \begin{minipage}{0.48\textwidth}
% \begin{tikzonimage}[trim= 10 15 10 10,clip,width=0.9\textwidth]{Figures_toy_model/toy_model_trajectory_01}%[tsx/show help lines]
% \node at (0.25,0.2) {\small $\textit{(b)}$};
% \end{tikzonimage}
% \begin{tikzonimage}[trim= 10 15 10 10,clip,width=0.9\textwidth]{Figures_toy_model/toy_model_trajectory_06}%[tsx/show help lines]
% \node at (0.25,0.2) {\small $\textit{(d)}$};
% \end{tikzonimage}
% \end{minipage}
% \caption{Panel (a) shows the average error defined in \eqref{eq:avg_error}. The remaining panels show the time history of the observed output $y(t)$ for different step inputs $u$. The label ``FOM'' indicates the full-order model.}
% \label{fig:toy_model_error}
% \end{figure}

The models were tested by generating $100$ trajectories with $u$ sampled uniformly at random from the interval $\left[0.01,\beta\right)$.
The results are shown in figure \ref{fig:toy_model}a, where the average error over trajectories is defined as
\begin{equation}
\label{eq:avg_error}
    e(t) = \frac{1}{N_\text{traj}}\sum_{j=0}^{N_{\text{traj}}-1} \frac{1}{\alpha_j} \lVert y^{(j)}(t) - \hat{y}^{(j)}(t)\rVert^2,
\end{equation}
with $\alpha_j$ as in \eqref{eq:cost_function}.
This figure indicates that NiTROM provides an accurate prediction of the time-history of the output $y(t)$ for all amplitudes $u \in [0,\beta)$.
In particular, NiTROM's performance is comparable with TrOOP's, and they both outperform (on average) the POD Galerkin and Operator Inference models. 
In order to demonstrate that NiTROM has learned a good representation of the dynamics on a two-dimensional subspace, we also test its predictive capabilities when the system is forced with sinusoidal inputs. 
(Notice that we did not train on sinusoids.)
Results for a representative response to a sinusoid are shown in figure \ref{fig:toy_model}b, where NiTROM and TrOOP are capable of accurately predicting the time history of the output $y(t)$, while the POD Galerkin and Operator Inference reduced-order models struggle to do so. 
For completeness, we also show the decay of the loss as a function of conjugate gradient iterations for both TrOOP and NiTROM in figure \ref{fig:cost_fnct_vals}.
In particular, we observe that while NiTROM eventually attains a slightly lower loss function value, we see that TrOOP reaches the stopping criterion $\lVert \nabla J \rVert \leq 10^{-6}$ much faster than NiTROM. 
Presumably, this is due to the fact that NiTROM's optimization landscape is ``less friendly'' than TrOOP's, as NiTROM admits a larger class of solutions. 
In fact, while the larger number of parameters in NiTROM allows for a wider class of reduced-order models, it may also make it more difficult for the optimizer to find a ``good" local minimum. 

\begin{figure}
\centering
\begin{tikzonimage}[trim= 10 80 10 110,clip,width=0.6\textwidth]{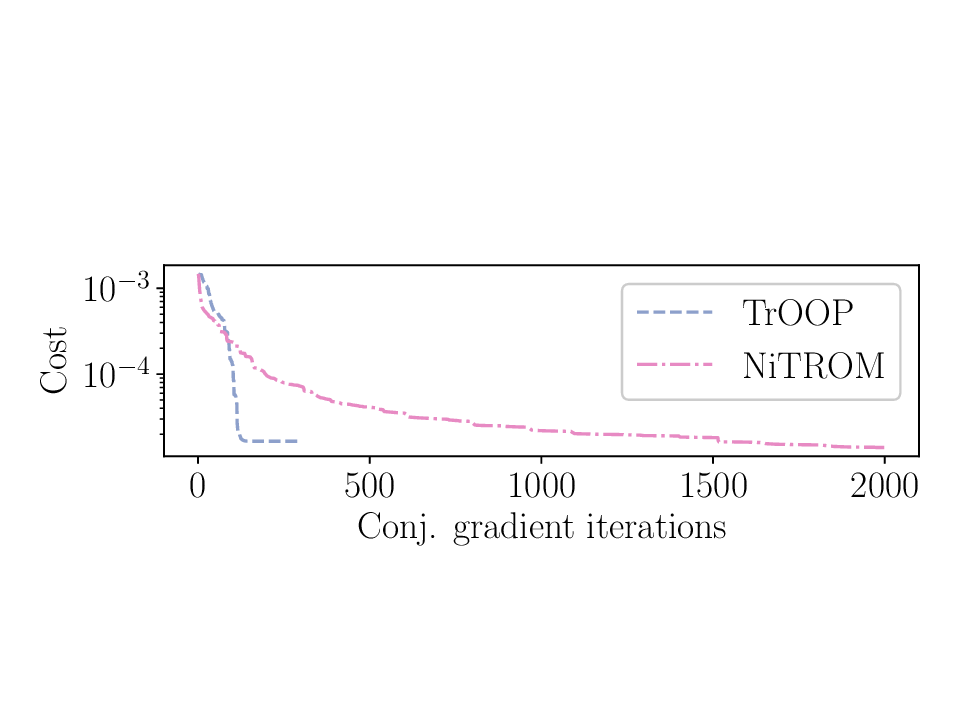}%[tsx/show help lines]
% \node at (0.19,0.4) {\small $\textit{(b)}$};
\end{tikzonimage}
\caption{Cost function value as a function of conjugate gradient iteration for the toy model.}
\label{fig:cost_fnct_vals}
\end{figure}

\section{Application to the complex Ginzburg-Landau (CGL) equation}
\label{sec:cgl}

In this section we consider the complex Ginzburg-Landau (CGL) equation 
\begin{equation}
    \label{eq:cgl} 
    \frac{\partial q}{\partial t} = \left(\nu\frac{\partial}{\partial x} +\gamma \frac{\partial^2}{\partial x^2} + \mu \right)q - a \lvert q\rvert^2 q,\quad x\in \left(-\infty,\infty\right),\,\, q(x,t)\in\mathbb{C},
\end{equation}
with parameters $a = 0.1$, $\nu = 2 + 0.4 i$, $\gamma = 1 - i$, $\mu = \left(\mu_0 - 0.2^2\right) + \mu_2 x^2/2$ with $\mu_0 = 0.38$ and $\mu_2 = -0.01$. 
Here, $i = \sqrt{-1}$.
For this choice of parameters, the origin $q(x,t)=0$ is linearly stable, but exhibits significant transient growth due to the non-normal nature of the linear dynamics \cite{ilak2010}. 
This type of behavior is common in high-shear flows (e.g., boundary layers, mixing layers and jets), making the CGL a meaningful and widely-used benchmark example. 
In this section, we are interested in computing reduced-order models capable of predicting the input-output dynamics of \eqref{eq:cgl} in response to spatially-localized inputs. 
In particular, we wish to predict the time history of complex-valued measurements
\begin{equation}
    y = Cq = \exp\bigg\{-\left(\frac{x + \overline{x}}{s}\right)^2\bigg\} q
\end{equation}
in response to complex-valued inputs $u$ that enter the dynamics according to 
\begin{equation}
    B u = \exp\bigg\{-\left(\frac{x - \overline{x}}{s}\right)^2\bigg\}u.
\end{equation}
Here, $s = 1.6$ and $\overline{x} =-\sqrt{-2\left(\mu_0 - 0.2^2\right)/\mu_2}$ is the location of the so-called ``branch I'' of the disturbance-amplification region (see \cite{ilak2010} for additional details). 
Upon spatial discretization on a grid with $N$ nodes, equation \eqref{eq:cgl} can be written as a real-valued dynamical system with cubic dynamics
\begin{equation}
\label{eq:cgl_discrete}
    \begin{aligned}
        \frac{\mathrm{d}\vq}{\mathrm{d}t} &= \mA \vq + \mH:\left(\vq \otimes \vq \otimes \vq\right) + \mB \vu \\ \vy &= \mC \vq,
    \end{aligned}
\end{equation}
where the state $\vq \in \mathbb{R}^{2N}$ contains the spatially-discretized real and imaginary components of~$q$, $\vu \in \mathbb{R}^2$ contains the real and imaginary components of the input $u$ and $\vy \in \mathbb{R}^2$ contains the real and imaginary components of the output $y$.
Thus, given the form of the full-order system, we seek cubic reduced-order models with dynamics expressed as the sum of a linear term, a cubic term and a linear input term. 

\begin{figure}
\centering
\begin{minipage}{0.48\textwidth}
\begin{tikzonimage}[trim= 10 15 10 10,clip,width=0.85\textwidth]{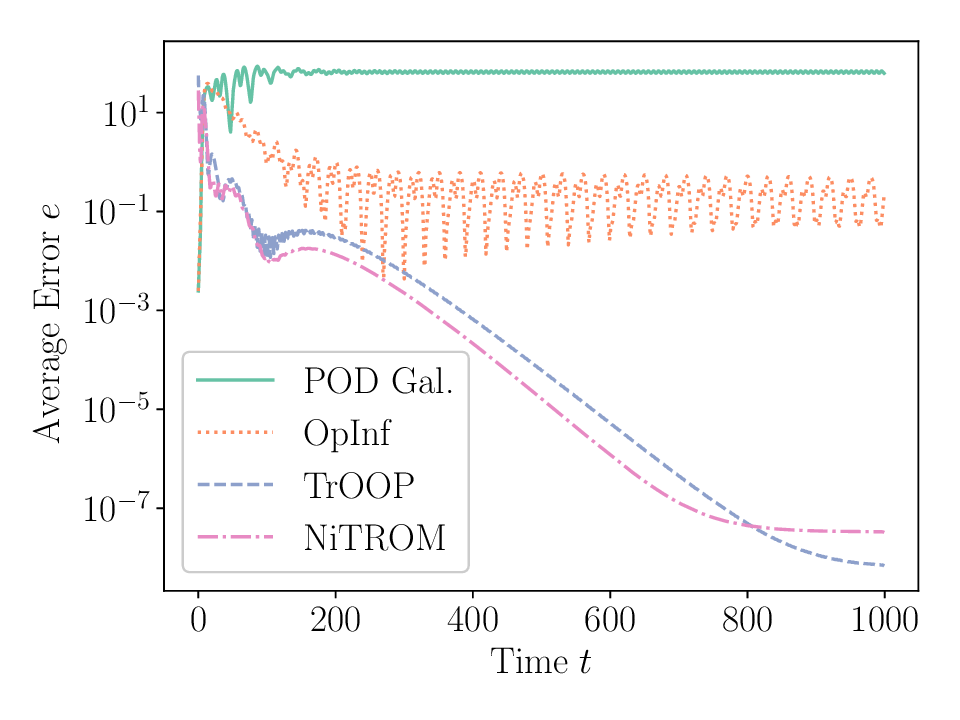}%[tsx/show help lines]
\node at (0.9,0.9) {\small $\textit{(a)}$};
\end{tikzonimage}
\end{minipage}
% \hspace{-4ex}
\begin{minipage}{0.48\textwidth}
\begin{tikzonimage}[trim= 10 15 10 10,clip,width=0.85\textwidth]{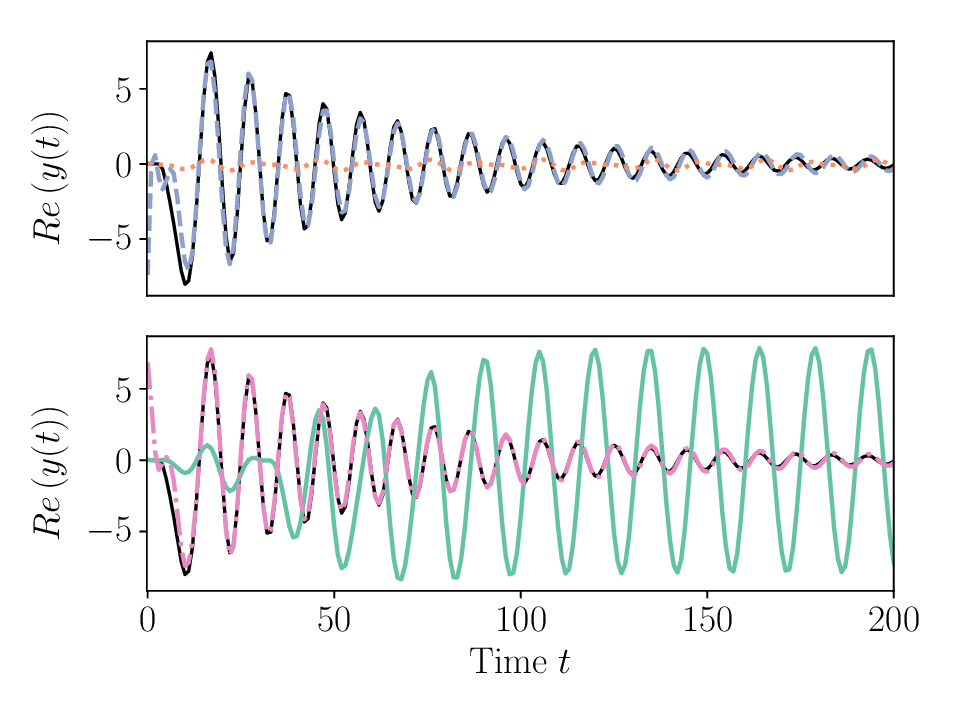}%[tsx/show help lines]
\node at (0.9,0.9) {\small $\textit{(b)}$};
\end{tikzonimage}
\end{minipage}
\caption{Panel (a) shows the average testing error defined in \eqref{eq:avg_error}. Panel (b) shows the real part of the output $y$ from a representative testing impulse response. The black line in panel (b) denotes the ground-truth response.}
\label{fig:cgl}
\end{figure}

\begin{figure}
\centering
\begin{minipage}{0.48\textwidth}
\begin{tikzonimage}[trim= 10 15 10 10,clip,width=0.85\textwidth]{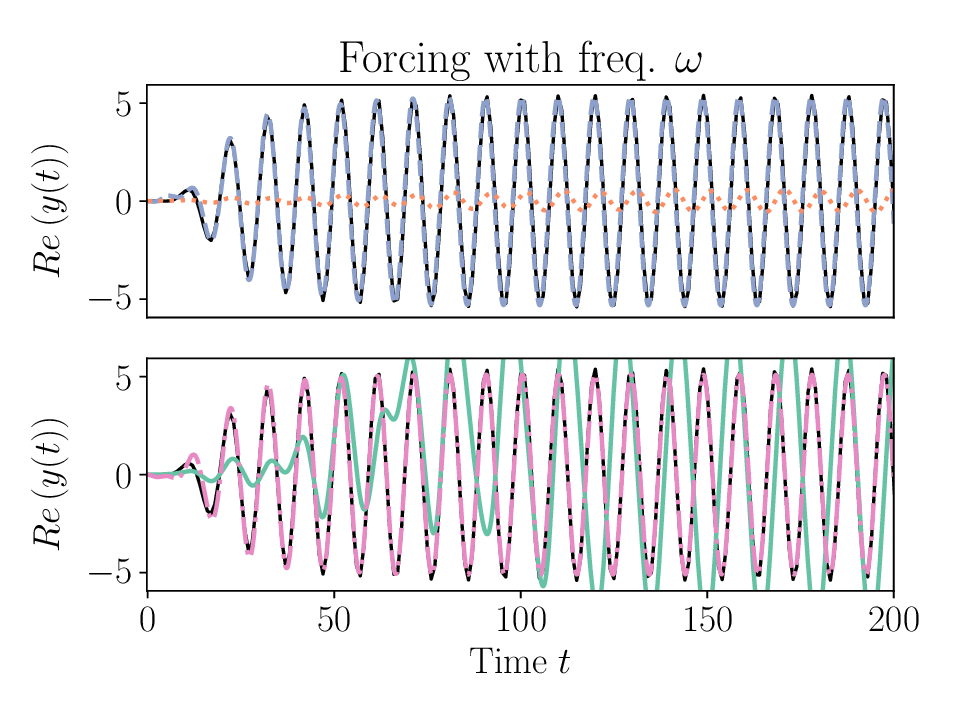}%[tsx/show help lines]
\node at (0.2,0.84) {\small $\textit{(a)}$};
\end{tikzonimage}
\end{minipage}
% \hspace{-4ex}
\begin{minipage}{0.48\textwidth}
\begin{tikzonimage}[trim= 10 15 10 10,clip,width=0.85\textwidth]{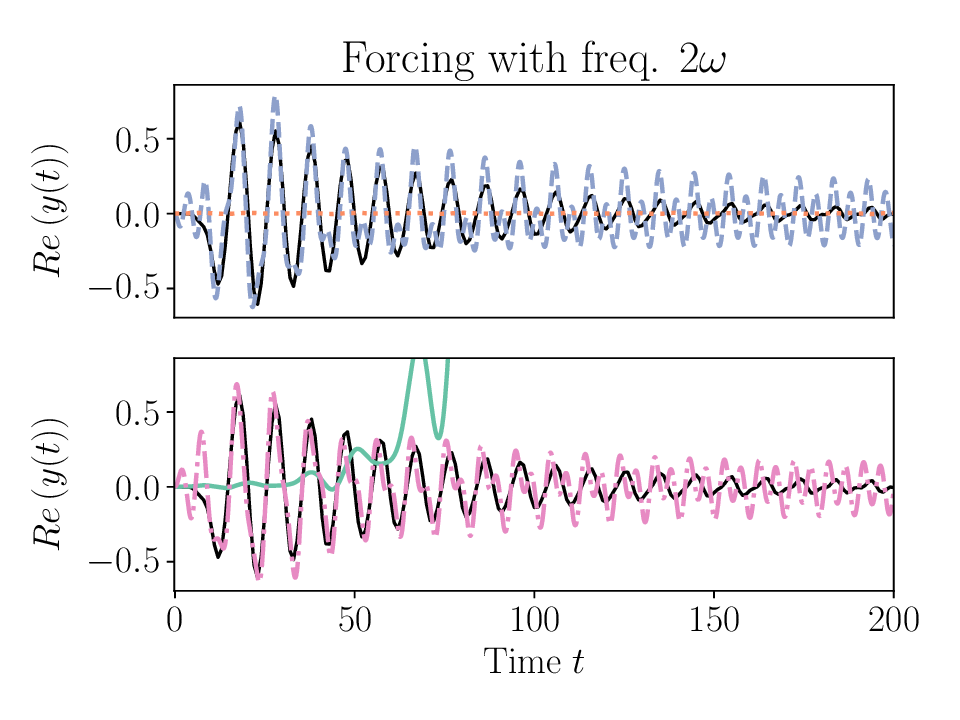}%[tsx/show help lines]
\node at (0.9,0.84) {\small $\textit{(b)}$};
\end{tikzonimage}
\end{minipage}
\caption{Real part of the output $y$ in response to a sinusoidal input with frequencies $\omega$ and $2\omega$, where $\omega$ is fundamental frequency of the system. The black continuous line indicates the ground truth, and the rest of the legend is in figure \ref{fig:cgl}a.}
\label{fig:cgl_sin}
\end{figure}

We train our models by simulating the response of \eqref{eq:cgl_discrete} to impulses $\beta \mB \ve_j$, where $\ve_j \in\mathbb{R}^2$ is the unit-norm vector in the standard basis, and $\beta \in \{-1.0,0.01,0.1,1.0\}$. 
We therefore have a total of $N_{\text{traj}} = 6$ training trajectories, and we collect the output $\vy$ at $N = 1000$ uniformly-spaced time instances $t_i \in [0,1000]$.
Since the leading five POD modes associated with the training data contain approximately $98\%$ of the variance and are sufficient to reconstruct the time-history of the output $\vy$ almost perfectly, we seek models of size $r = 5$.
% Throughout, we consider models of size $r = 5$, since the leading five POD modes of the training data set contain approximately $99.8\%$ of the variance, and five modes were sufficient to reconstruct the time-history of the output $\vy$ almost perfectly. 
The cost functions for NiTROM, TrOOP and Operator Inference are analogous to those considered in section~\ref{sec:toy_model}, except that the reduced-order dynamics are cubic and the normalization constants $\alpha_j$ in \eqref{eq:cost_function} are defined as the time-averaged energy of the output $\vy$ along the $j$th trajectory.
In Operator Inference, the regularization parameter for the reduced-order fourth-order tensor was chosen as $\lambda = 10^9$ following the same criterion described in the previous section. 
The NiTROM optimization was initialized with $\mPhi = \mPsi$ given by the first five POD modes of the training data and the reduced-order tensors provided by Operator Inference. 
The optimization was conducted using \emph{coordinate} descent by successively holding the reduced-order tensors fixed and allowing for the bases $\mPhi$ and $\mPsi$ to vary, and viceversa. 
On this particular example, we found this procedure to be less prone to getting stuck in ``bad'' local minima. 
% \bvremark{do you anticipate NiTROM solution space searches having lots of problems like this?}  
TrOOP, on the other hand, was initialized with $\mPhi$ and $\mPsi$ given by Balanced Truncation \cite{Moore1981principal,Rowley2005model} since the initialization with POD modes led to a rather inaccurate local minimum. 
TrOOP's optimization was carried out using conjugate-gradient.

\begin{figure}
\centering
\begin{tikzonimage}[trim= 10 80 10 110,clip,width=0.6\textwidth]{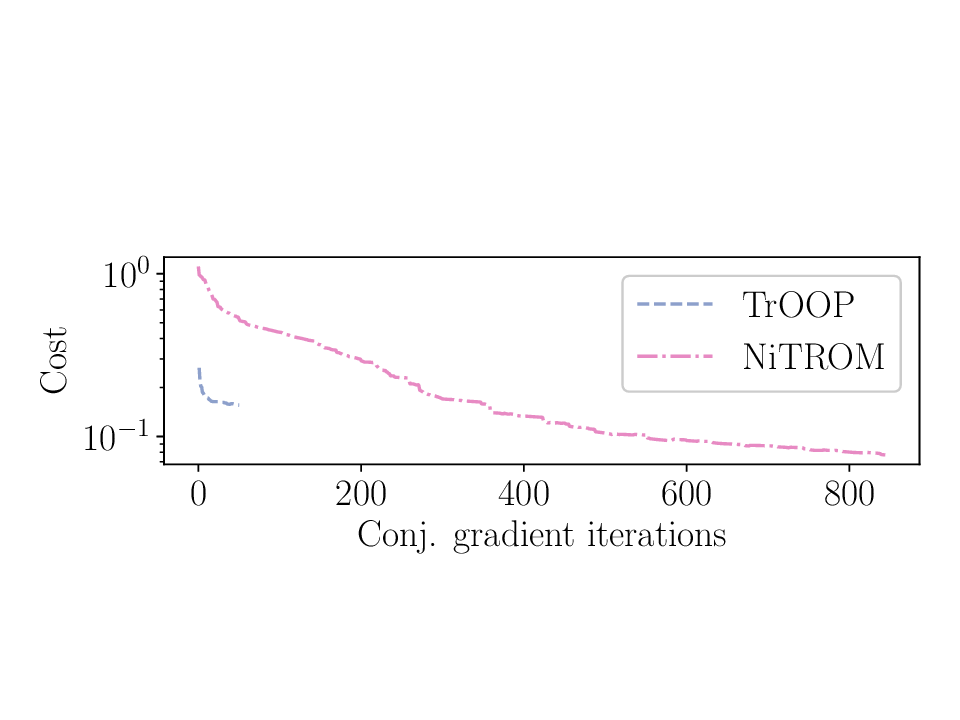}%[tsx/show help lines]
% \node at (0.19,0.4) {\small $\textit{(b)}$};
\end{tikzonimage}
\caption{Cost function value as a function of conjugate gradient iteration for the CGL equation. TrOOP was initialized using Balanced Truncation, while NiTROM using Operator Inference.}
\label{fig:cost_fnct_vals_cgl}
\end{figure}

We test the performance of our model by generating 50 impulse responses $\beta \mB \ve_j$ with $\beta$ drawn uniformly at random from $[-1.0,1.0]$. 
The average error across all testing trajectories is shown in figure \ref{fig:cgl}a, while a representative impulse response is shown in figure \ref{fig:cgl}b. 
Overall, we see that both NiTROM and TrOOP achieve very good predictive accuracy and are capable of tracking the output through the heavy oscillatory transients. 
By contrast, Operator Inference and the POD-Galerkin model exhibit higher errors, and this is most likely due to the highly non-normal nature of the CGL dynamics. 
In fact, both these methods achieve dimensionality reduction by \emph{orthogonally} projecting the state onto the span of POD modes, while, as previously discussed, reduced-order models for non-normal systems typically require carefully chosen oblique projections. 
Finally, we demonstrate the predictive accuracy of NiTROM on unseen sinusoidal inputs of the form $0.05 \sin(k\omega t)\mB\vu/\lVert \mB\vu\rVert$, where $\vu \in\mathbb{R}^2$ is chosen at random and $\omega \approx 0.648$ is the natural frequency of the system.
The results for frequencies $\omega$ and $2\omega$ are shown in figure \ref{fig:cgl_sin}, where we see that NiTROM provides an accurate estimate of the response of the system at frequency $\omega$ and an acceptable prediction at frequency $2\omega$. 
The reason why the prediction at $2\omega$ for both TrOOP and NiTROM is not as clean as the prediction at $\omega$ is because the training data exhibited dominant oscillatory dynamics at the natural frequency $\omega$ and very little contributions from other frequencies.
Nonetheless, the predictions at $2\omega$ are better than those provided by POD-Galerkin and Operator Inference.
% Therefore, the training process led the models to inherit the sensitivity of the full-order model to frequency-$\omega$ perturbations and to neglect the low-energy oscillations at all other frequencies. 
Before closing this example, we report on the loss function value for both TrOOP and NiTROM in figure \ref{fig:cost_fnct_vals_cgl}, but we keep in mind that TrOOP was initialized using Balanced Truncation, while NiTROM was initialized using Operator Inference.

\section{Application to the lid-driven cavity flow}

\begin{figure}
\centering
\begin{minipage}{0.48\textwidth}
\begin{tikzonimage}[trim= 40 15 20 15,clip,width=0.85\textwidth]{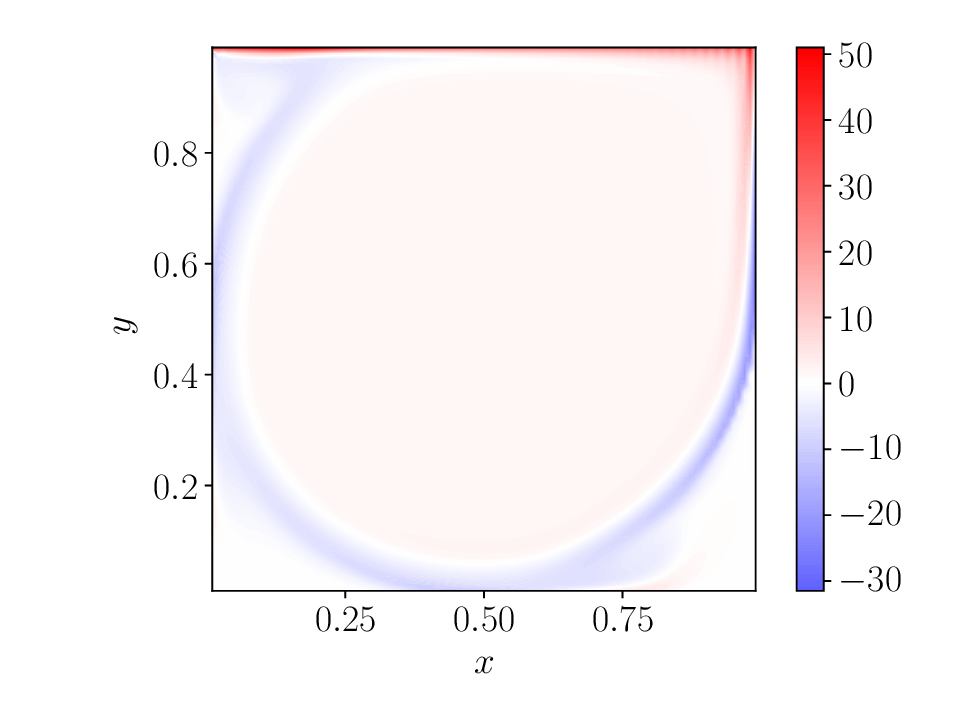}%[tsx/show help lines]
\node at (0.75,0.23) {\small $\textit{(a)}$};
\end{tikzonimage}
\end{minipage}
% \hspace{-4ex}
\begin{minipage}{0.48\textwidth}
\begin{tikzonimage}[trim= 10 10 5 5,clip,width=0.85\textwidth]{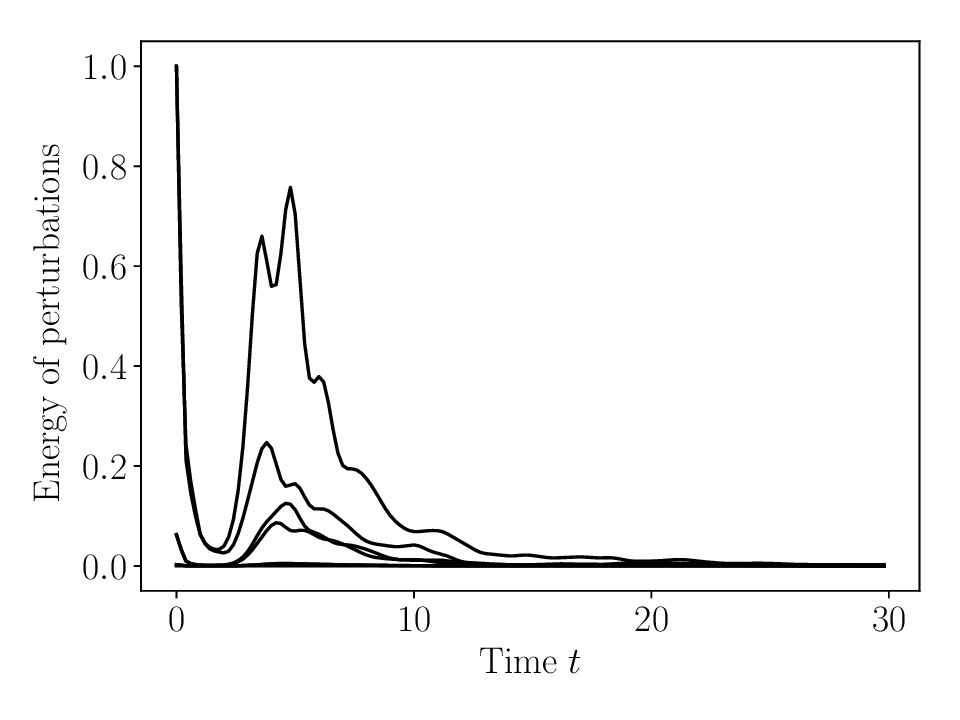}%[tsx/show help lines]
\node at (0.9,0.245) {\small $\textit{(b)}$};
\end{tikzonimage}
\end{minipage}
\caption{Panel (a) shows the vorticity field from the steady-state solution that exists at $Re = 8300$, and panel (b) shows the energy (i.e., the squared two norm) of the seven training trajectories.}
\label{fig:cavity}
\end{figure}

In this section, we apply our model reduction procedure to an incompressible fluid flow inside a lid-driven square cavity. 
The flow dynamics are governed by the incompressible Navier-Stokes equation and by the continuity equation 
\begin{align}
    \frac{\partial \vu}{\partial t} + \vu \cdot\nabla \vu &= -\nabla p + Re^{-1}\nabla^2\vu \\ 
    \nabla\cdot \vu &= 0, \label{eq:continuity}
\end{align}
where $\vu(\vx,t) = \left(u(\vx,t),v(\vx,t)\right)$ is the two-dimensional velocity vector, $p(\vx,t)$ is the pressure and $Re$ is the Reynolds number. 
Throughout, we consider a two-dimensional spatial domain $D = [0,1]\times [0,1]$ with zero-velocity boundary conditions at all walls, except for $u = 1$ at the top wall. 
The Reynolds number is held at $Re = 8300$, where the flow admits a linearly stable steady state (shown in figure \ref{fig:cavity}a), but exhibits large amplification and significant transient growth due to the non-normal nature of the underlying linear dynamics. 
The high degree of non-normality and consequent transient growth can be appreciated by looking at figure \ref{fig:cavity}b, where we show the time-history of the energy of several impulse responses. 
In particular, we see that after an initial decay, the energy spikes around $t  = 5$ before decaying back to zero. 
We discretize the governing equations using a second-order finite-volume scheme on a uniform fully-staggered grid of size $N_x \times N_y = 100 \times 100$. 
With this spatial discretization, no pressure boundary conditions need to be imposed. 
The temporal integration is carried out using the second-order fractional step (projection) method introduced in \cite{chorin68}.
Our solver was validated by reproducing some of the results in \cite{ghia1982}.

In this example, we are interested in computing data-driven reduced-order models capable of predicting the evolution of the flow in response to spatially-localized inputs that enter the $x$-momentum equation as
\begin{equation}
\label{eq:Bw}
    B(x,y) w(t) = \exp\bigg\{-5000\left(\left(x - x_c\right)^2 + \left(y - y_c\right)^2\right)\bigg\} w(t),
\end{equation}
with $x_c = y_c = 0.95$.
Upon spatial discretization and removal of the pressure via projection onto the space of divergence-free vector fields, the dynamics are governed by 
\begin{equation}
\label{eq:cavity_disc}
    \frac{\mathrm{d}}{\mathrm{d}t} \vq = \mA \vq + \mH\left(\vq\otimes \vq\right) + \mB w,
\end{equation}
where $\vq \in \mathbb{R}^{N}$ is the spatially-discretized divergence-free velocity field (with $N = 2N_x N_y = 2\times 10^{4}$), $\mA$ governs the linear dynamics, $\mH$ is a third-order tensor representative of the quadratic nonlinearity in the Navier-Stokes equation and $\mB$ is the input matrix obtained from~\eqref{eq:Bw} after enforcing that $\mB$ generates a divergence-free vector.
(For convenience, we also scale $\mB$ to unit norm.)
Throughout the remainder of this section, we take $\vy = \vq$ (i.e., we observe the time evolution of the whole state).

\subsection{Training procedure}

We collect seven training trajectories by generating impulses responses of \eqref{eq:cavity_disc} with $w \in \{-1.0,-0.25,-0.05,0.01,0.05,0.25,1.0\}$.
The time-history of their energy is shown in figure \ref{fig:cavity}b.
We save 150 snapshots from each trajectory at equally-distributed temporal instances $t \in [0,30]$, and then we perform POD. 
Using the first 50 POD modes, which contain $99.4\%$ of the variance in the training data, we compute an Operator Inference model by minimizing the cost function \eqref{eq:opt_problem_opinf}.
We normalize the trajectories by their time-averaged energy and, as in the previous sections, we also penalize the Frobenius norm of the third-order tensor $\mH$ with the regularization parameter take to be $\lambda = 10^{-2}$.

Given the complexity of the problem and the length of the trajectories, we train NiTROM as follows. 
First, we pre-project the data onto the span of the first 200 POD modes, which contain $> 99.99\%$ of the variance. 
This guarantees that the optimal NiTROM bases~$\mPhi$ and~$\mPsi$ satisfy the divergence-free constraint in \eqref{eq:continuity}, since the POD modes are computed from divergence-free snapshots.
Second, after initializing the search with the Operator Inference model, we train by progressively extending the length of the forecasting horizon. 
That is, we first optimize a model to make predictions up to $t = 5$, then $t = 10$, and so forth all the way up to $t = 30$. 
% \bvremark{is this a good way in general of avoiding local minima?}
Since, after a first pass, our model exhibited slightly unstable linear dynamics (possibly due to the presence of numerical noise and/or weak decaying oscillations in the tail end of the training data), we extended the training horizon all the way up to $t = 60$ and added a stability-promoting penalty to our cost function as follows, 
\begin{equation}
    \widetilde{J} = J_{\text{NiTROM}} + \mu \lVert \hat{\vz}_\text{lin}(t_f)\rVert^2.
\end{equation}
Here, $t_f$ is a sufficiently large time (chosen to be $100$ in our case) and $\hat{\vz}_{\text{lin}}$ satisfies 
\begin{equation}
    \frac{\mathrm{d} \hat\vz_{\text{lin}}}{\mathrm{d} t} = \mA_r \hat\vz_{\text{lin}},\quad \hat\vz_{\text{lin}}(0) = \hat\vz_{\text{lin},0},
\end{equation}
with $\hat\vz_{\text{lin},0}$ a unit-norm random vector. 
Notice that this penalty is truly stability-promoting, as it is analogous to penalizing the Frobenius norm of $e^{\mA_r t_f}$, and shrinking the Frobenius norm of the exponential map corresponds to pushing the eigenvalues of $\mA_r$ farther into the left-half plane.
% This penalty is designed so that if $\mA_r$ has eigenvalues in the right-half plane, the penalty term will dominate the cost function $\widetilde{J}$ and gradient-descent should then prioritize models with linearly-stable dynamics (i.e., with $\lVert \hat{\vz}_\text{lin}(t_f)\rVert^2 \ll 1$).
The gradient of the penalty term with respect to $\mA_r$ can be computed straightforwardly following the same logic used in Proposition \ref{prop:gradient}.
The regularization parameter $\mu$ was held at zero for most of the training, until we reached a forecasting horizon $t = 60$ when we set $\mu = 10^{-3}$.
The training was conducted using coordinate descent as described in section \ref{sec:cgl}, and we stopped the optimization after approximately 2000 iterations.

\begin{figure}
\centering
\begin{minipage}{0.48\textwidth}
\begin{tikzonimage}[trim= 10 10 5 5,clip,width=0.85\textwidth]{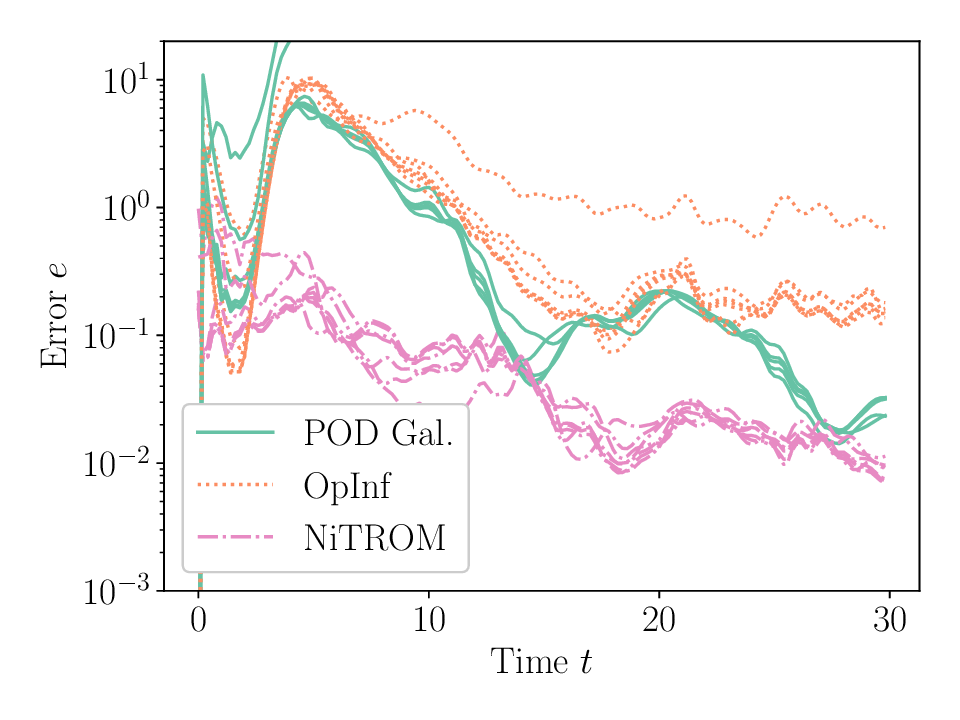}%[tsx/show help lines]
\node at (0.9,0.22) {\small $\textit{(a)}$};
\end{tikzonimage}
\end{minipage}
% \hspace{-4ex}
\begin{minipage}{0.48\textwidth}
\begin{tikzonimage}[trim= 10 10 5 5,clip,width=0.85\textwidth]{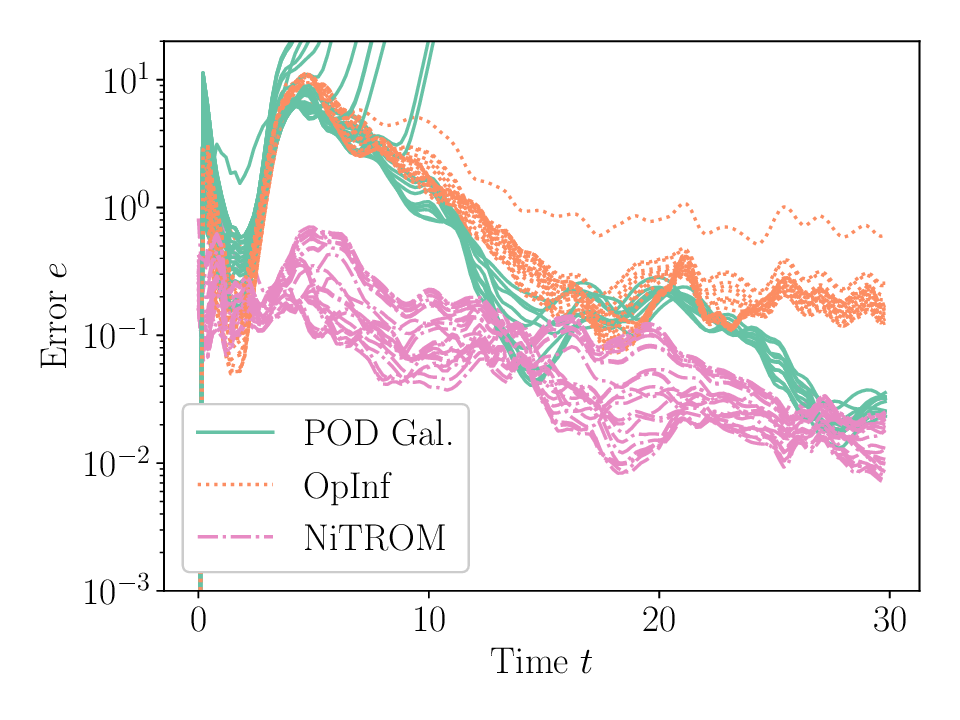}%[tsx/show help lines]
\node at (0.9,0.22) {\small $\textit{(b)}$};
\end{tikzonimage}
\end{minipage}
\caption{Panel (a) shows the training error from the 7 training impulses responses, and panel (b) shows the testing error computed for 25 unseen impulse responses. The error is defined in equation \eqref{eq:error_cavity}.}
\label{fig:cavity_errors}
\end{figure}

\begin{figure}
\centering
\begin{minipage}{0.48\textwidth}
\begin{tikzonimage}[trim= 5 10 5 5,clip,width=0.85\textwidth]{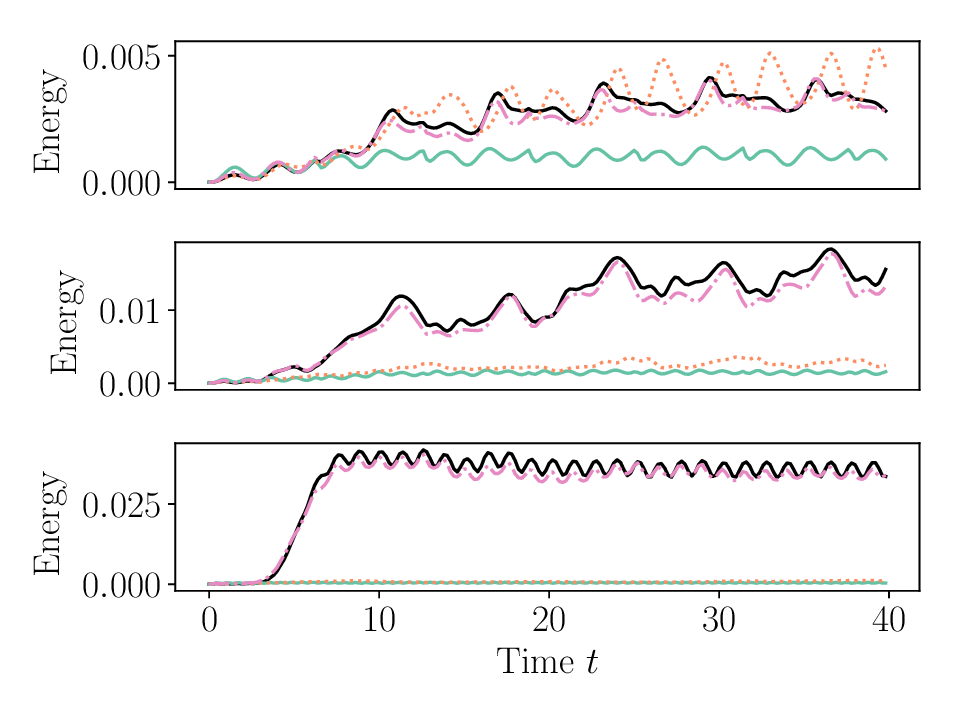}%[tsx/show help lines]
\node at (0.575,1.0) {\small $w(t) = 0.1\sin(1.25 k t)$};
\node at (0.25,0.9) {\small $k = 1$};
\node at (0.25,0.6) {\small $k = 2$};
\node at (0.25,0.31) {\small $k = 4$};
\end{tikzonimage}
\end{minipage}
% \hspace{-4ex}
\begin{minipage}{0.48\textwidth}
\begin{tikzonimage}[trim= 5 10 5 5,clip,width=0.85\textwidth]{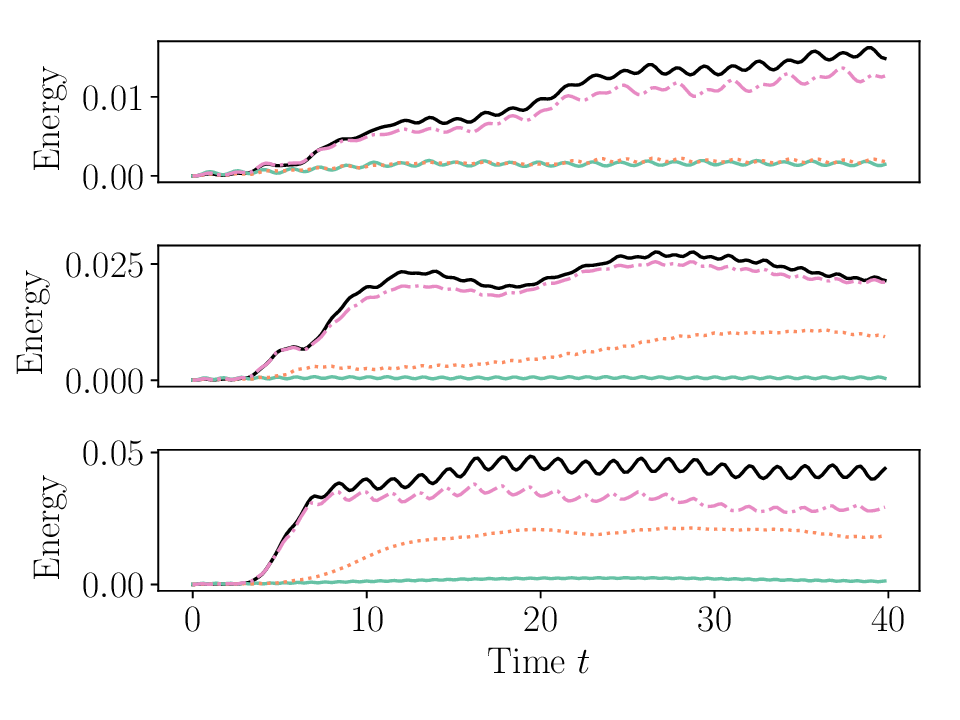}%[tsx/show help lines]
\node at (0.575,1.0) {\small $w(t) = 0.1\sin(k t)$};
\node at (0.235,0.9) {\small $k = 2$};
\node at (0.235,0.6) {\small $k = 3$};
\node at (0.235,0.31) {\small $k = 4$};
\end{tikzonimage}
\end{minipage}
\caption{Evolution of the energy of the perturbations in response to sinusoidal inputs $w(t)$. The black line is the full-order model and the rest of the legend is in figure \ref{fig:cavity_errors}.}
\label{fig:cavity_freq_response}
\end{figure}

\begin{figure}
\centering
\subfloat{
    \begin{tikzonimage}[trim= 45 15 80 6,clip,width=0.4\textwidth]{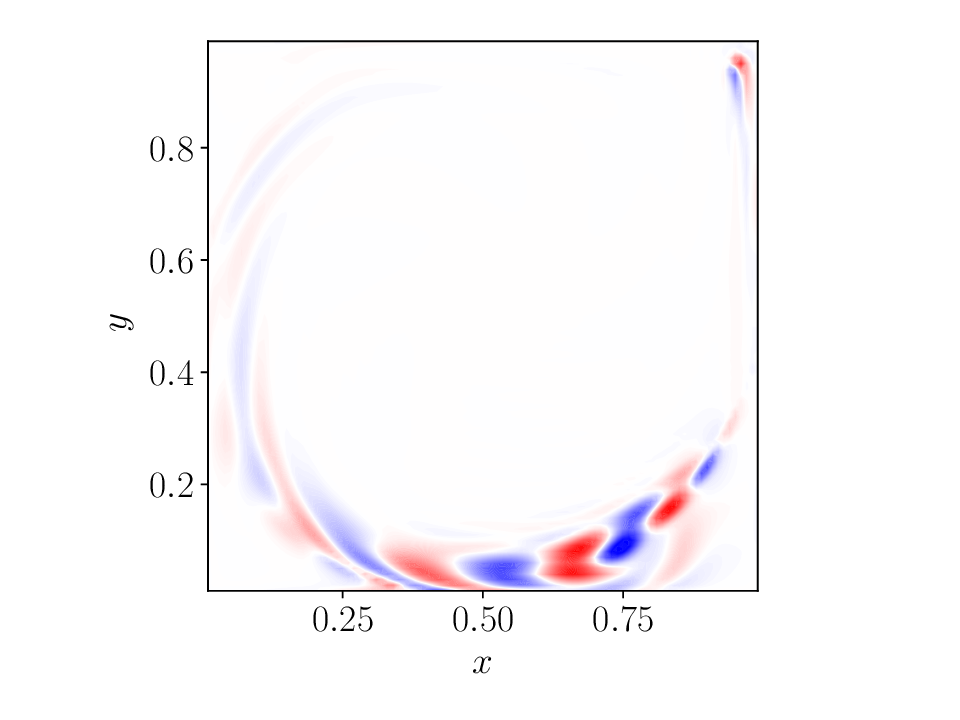}%[tsx/show help lines]
    \node at (0.55,0.55) {\small FOM};
    \end{tikzonimage}
}
\hspace{4ex}
\subfloat{
    \begin{tikzonimage}[trim= 45 15 80 6,clip,width=0.4\textwidth]{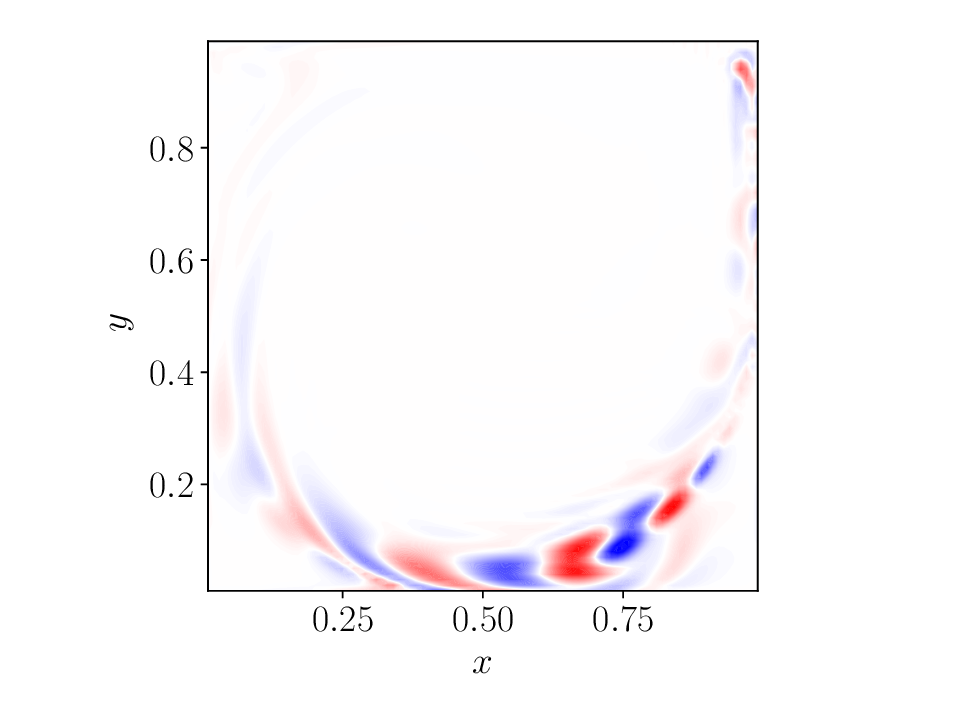}%[tsx/show help lines]
    \node at (0.55,0.55) {\small NiTROM};
    \end{tikzonimage}
}\\
\subfloat{
    \begin{tikzonimage}[trim= 45 15 80 6,clip,width=0.4\textwidth]{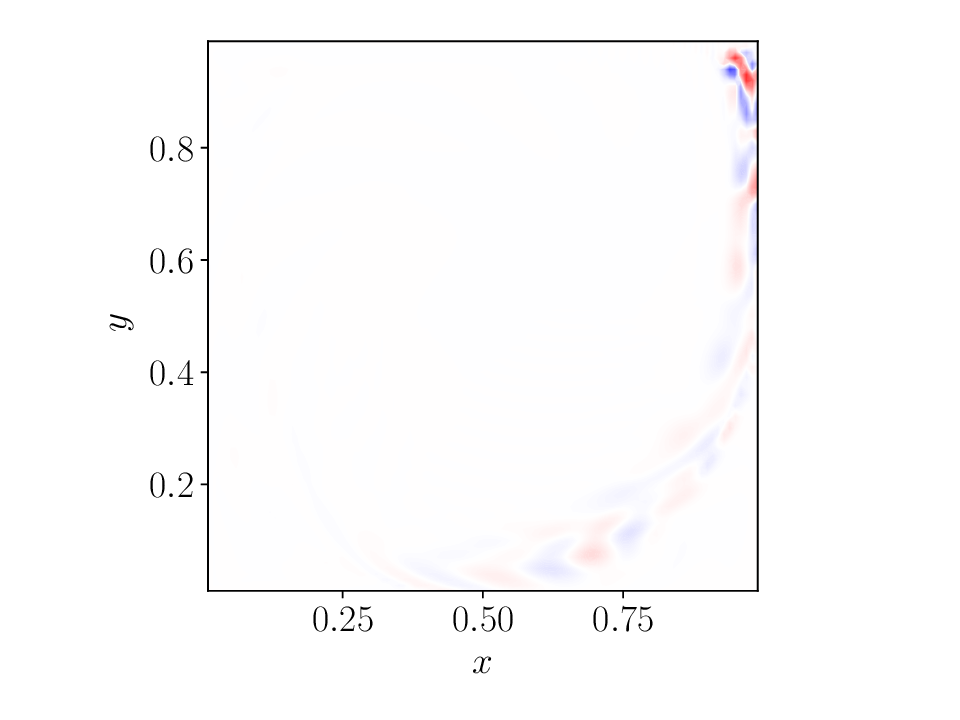}%[tsx/show help lines]
    \node at (0.55,0.55) {\small OpInf};
    \end{tikzonimage}
}
\hspace{4ex}
\subfloat{
    \begin{tikzonimage}[trim= 45 15 80 6,clip,width=0.4\textwidth]{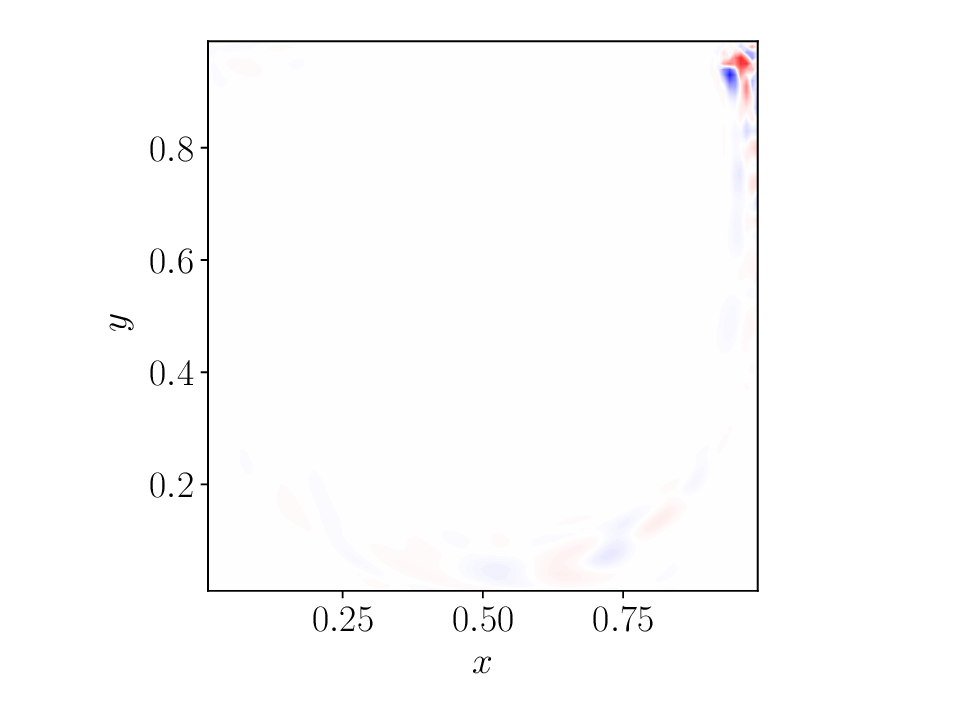}%[tsx/show help lines]
    \node at (0.55,0.55) {\small POD Gal.};
    \end{tikzonimage}
}
\caption{Vorticity field at time $t = 30$ from the trajectory in the bottom left panel of figure~\ref{fig:cavity_freq_response}. Red indicates positive vorticy with maximum value $0.7$, blue indicates negative vorticity with minimum value $-0.7$ and white is zero vorticity.}
\label{fig:cavity_snaps_1}
\end{figure}

\begin{figure}
\centering
\subfloat{
    \begin{tikzonimage}[trim= 45 15 80 6,clip,width=0.4\textwidth]{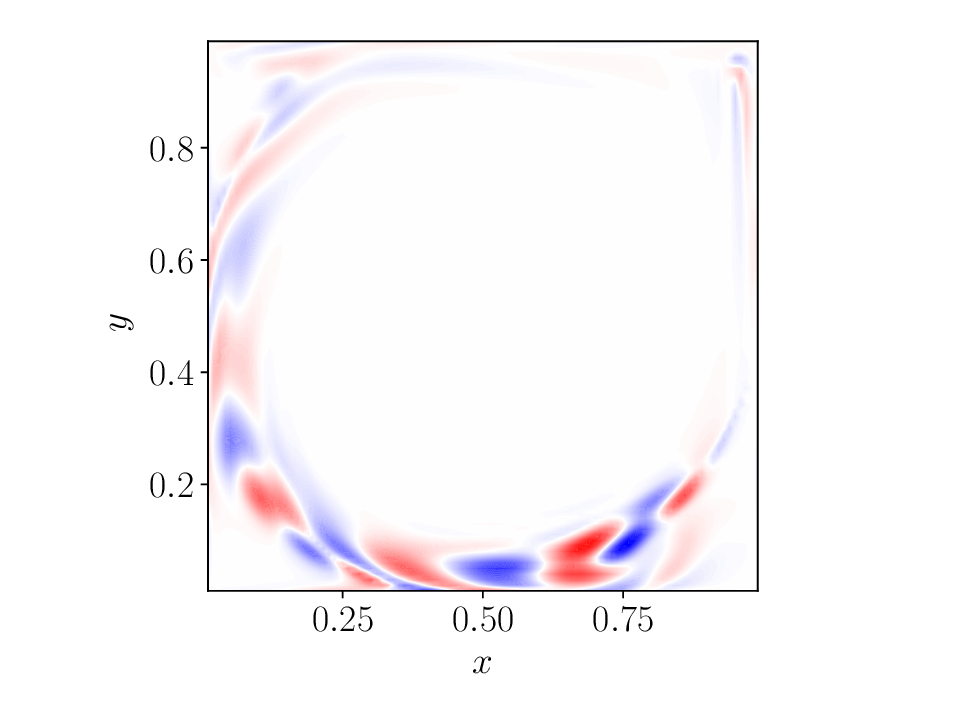}%[tsx/show help lines]
    \node at (0.55,0.55) {\small FOM};
    \end{tikzonimage}
}
\hspace{4ex}
\subfloat{
    \begin{tikzonimage}[trim= 45 15 80 6,clip,width=0.4\textwidth]{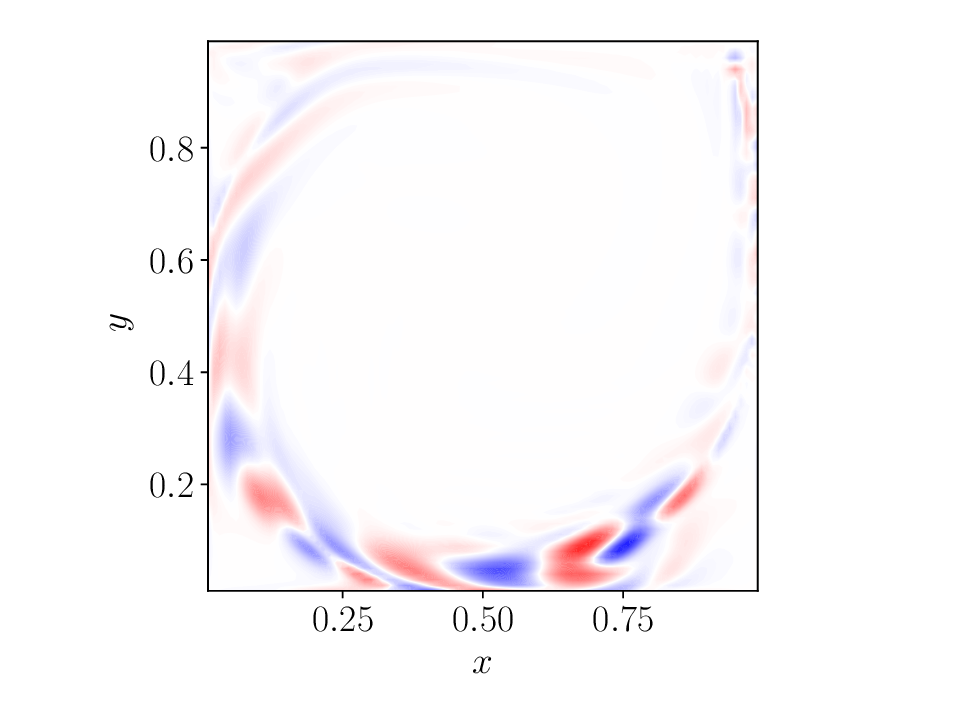}%[tsx/show help lines]
    \node at (0.55,0.55) {\small NiTROM};
    \end{tikzonimage}
}\\
\subfloat{
    \begin{tikzonimage}[trim= 45 15 80 6,clip,width=0.4\textwidth]{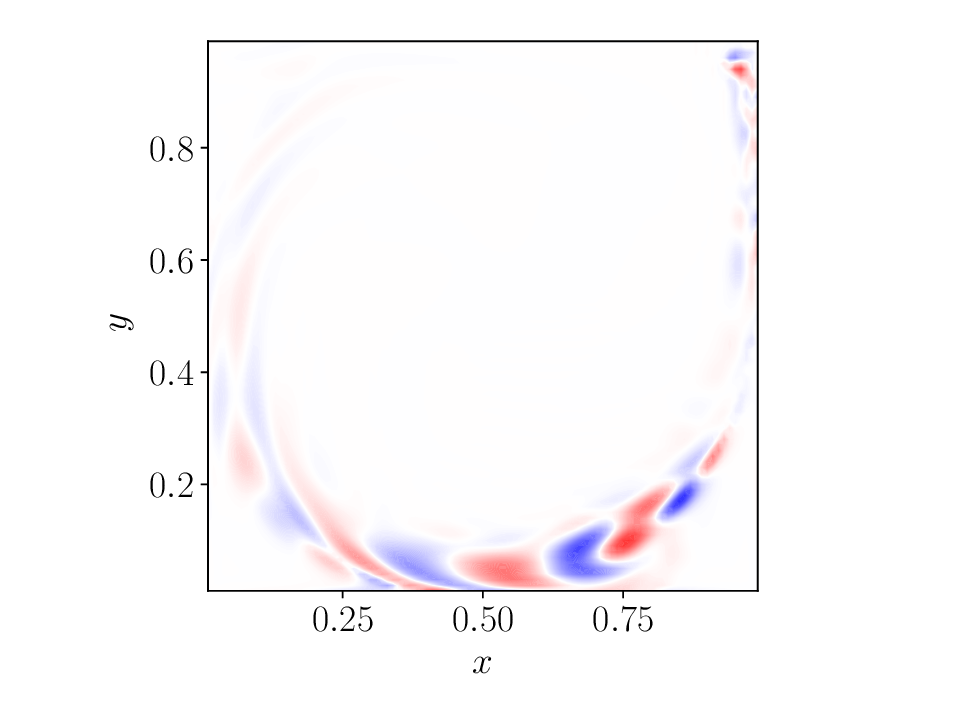}%[tsx/show help lines]
    \node at (0.55,0.55) {\small OpInf};
    \end{tikzonimage}
}
\hspace{4ex}
\subfloat{
    \begin{tikzonimage}[trim= 45 15 80 6,clip,width=0.4\textwidth]{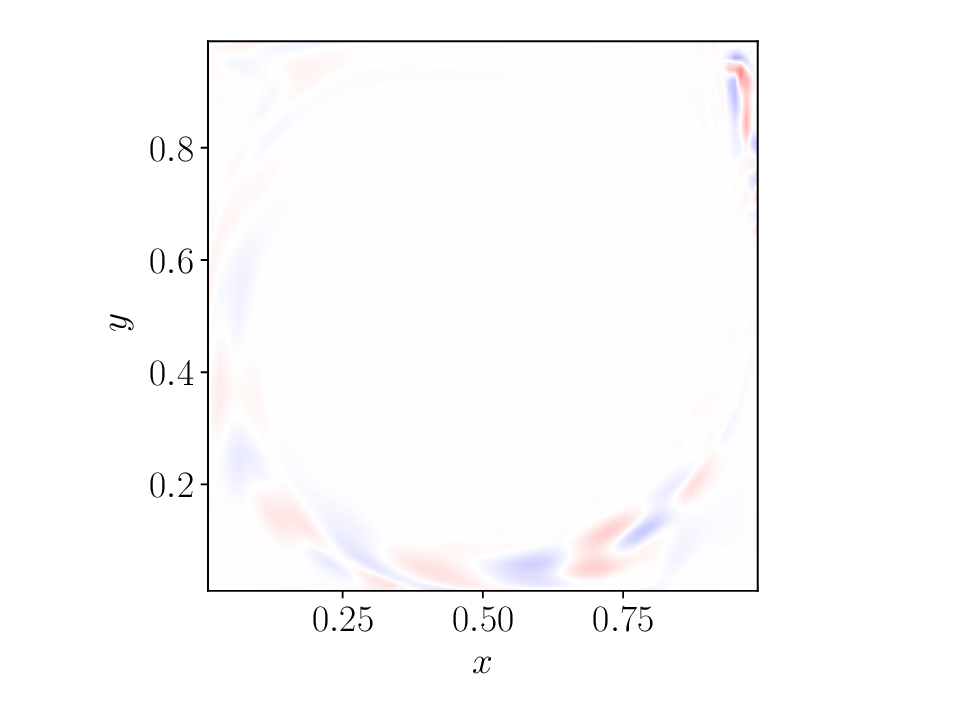}%[tsx/show help lines]
    \node at (0.55,0.55) {\small POD Gal.};
    \end{tikzonimage}
}
\caption{Vorticity field at time $t = 30$ from the trajectory in the bottom right panel of figure~\ref{fig:cavity_freq_response}. Red indicates positive vorticy with maximum value $0.75$, blue indicates negative vorticity with minimum value $-0.75$ and white is zero vorticity.}
\label{fig:cavity_snaps_2}
\end{figure}

\subsection{Testing}

In this section we compare NiTROM against Operator Inference and POD Galerkin.
We do not compare against TrOOP because of its intrusive need to access the linearized dynamics and the adjoint, and because we are ultimately interested in comparing our formulation against other non-intrusive (or weakly intrusive) model reduction techniques. 
We test the models by generating 25 impulse responses with the input $w$ drawn uniformly at random from $[-1,1]$. 
The training and testing errors for NiTROM, Operator Inference and for the POD-Galerkin model (all with dimension $r = 50$) are shown in figure \ref{fig:cavity_errors}.
The error is defined as
\begin{equation}
\label{eq:error_cavity}
    e(t) = \frac{1}{\sum_{i=0}^{N-1}\lVert \vq(t_i)\rVert^2} \sum_{i=0}^{N-1} \lVert \vq(t_i) - \hat{\vq}(t_i)\rVert^2,
\end{equation}
where $\vq$ is the ground-truth and $\hat{\vq}$ is the prediction given by the reduced-order model.
From the figure, we see that NiTROM maintains a low error across all trajectories and for all times. 
In particular, we observe that around $t = 5$ (when the fluid exhibits its peak in transient growth, as illustrated in figure \ref{fig:cavity}b) the errors produced by POD Galerkin and Operator Inference can be one to two orders of magnitude larger than those produced by NiTROM. 

As in the previous section, we also test the ability of our reduced-order model to predict the response of the fluid to sinusoidal inputs $w(t) = 0.1\sin(k\omega t)$ starting from the stable steady state. 
The results are shown in figure \ref{fig:cavity_freq_response}, where we see the response to harmonics of $\omega = 1.25$ and $\omega = 1$, which are frequencies that are naturally excited by the linear dynamics of the flow.
In all cases, NiTROM exhibits better predictive accuracy than the other models, and it is capable of tracking the early-stage sharp growth of the perturbations as well as the cavity's long-time oscillatory behavior. 
Finally, in order to gain further insight into the performance of these models, we show vorticity snapshots drawn from the two trajectories with frequency $4\omega$ at time $t = 30$.
In figure \ref{fig:cavity_snaps_1}, Operator Inference and POD Galerkin underestimate the magnitude of the vorticity field as well as the shape and phase of the vortical structures. 
In figure \ref{fig:cavity_snaps_2}, on the other hand, Operator Inference provides a reasonable approximation of the vorticity magnitude, but it misrepresents the phase of the vortical structures (notice the difference in colors at any given physical location). 
By contrast, NiTROM provides an accurate estimate of the vorticity phase and magnitude in both cases.

\section{Conclusion}

We have introduced a novel non-intrusive, purely data-driven framework to compute accurate reduced-order models of high-dimensional systems that exhibit non-normal behavior. 
In particular, given trajectories from the full-order system, we solve an optimization problem to simultaneously find optimal oblique projection operators and reduced-order dynamics. 
The framework is termed NiTROM---``Non-intrusive Trajectory-based optimization of Reduced-Order Models''---and demonstrated on three examples: a simple toy model governed by three ordinary differential equations, the complex Ginzburg-Landau equations and a two-dimensional incompressible lid-driven cavity flow at Reynolds number $Re = 8300$.
All the examples showcase the effectiveness of our models in giving accurate short- and long-time predictions even for operating regimes outside those in the training data. 
% The examples considered in the paper demonstrate the effectiveness of our method in learning low-order models capable of giving accurate predictions even for operating regimes outside those in the training data. 
In the future, it would be interesting to explore the possibility of extending our NiTROM formulation to quadratic manifolds, as done within the Operator Inference formulation in \cite{geelen2023}. 
% \bvremark{make note of comparison in formulation to Troop and OpInf?}

\section*{Acknowledgments}

This material is based upon work supported by the National Science Foundation under Grant No. 2139536, issued to the University of Illinois at Urbana-Champaign by the Texas Advanced Computing Center under subaward UTAUS-SUB00000545 with Dr. Daniel Stanzione as the PI.  
The computations for the lid-driven cavity flow were performed on TACC’s Frontera under LRAC grant CTS20006.

\appendix
\section{Proof of Proposition \ref{prop:gradient}}
\label{app:proof}
We begin by deriving the reduced-order adjoint equation \eqref{eq:adjoint_rom}.
The derivative of $L$ with respect to $\hat{\vz}$ in the direction of $\vxi$ is given by 
\begin{equation}
    \begin{aligned}
        D_{\hat\vz} L [\vxi] = \sum_{i=0}^{N-1}&\bigg\{-2\ve(t_i)^\intercal \mC(t_i)\mPhi \left(\mPsi^\intercal \mPhi\right)^{-1}\vxi(t_i) + \vlam_i^\intercal \vxi \bigg\vert_{t_0}^{t_i} - \int_{t_0}^{t_i}\left(\frac{d\vlam^\intercal_i}{dt} + \vlam_i^\intercal \left[\partial_{\hat\vz}\overline{\vf}_r(\hat\vz)\right]\right)\vxi dt\bigg. \\
        &\bigg. +\vlam_i(t_0)^\intercal \vxi(t_0)\bigg\},
    \end{aligned}
\end{equation}
where we have used integration by parts on the time-derivative term. 
For optimality, we require $D L[\vxi] = 0$ for all $\vxi$. 
Thus, for $i > 0$, the terms $\vlam_i(t_0)^\intercal \vxi(t_0)$ cancel out and we are left with equation \eqref{eq:adjoint_rom}. 
Similarly, when $i = 0$, the second and third terms in the sum vanish and we recover 
\begin{equation}
    \vlam_0(t_0) = 2 \minvT \mPhi^\intercal \mC(t_0)^\intercal \ve(t_0).
\end{equation}
We now derive the gradient of $L$ with respect to $\mPhi$. 
% First, recalling the ambient space metric in \eqref{eq:Phi_metric}, the gradient is defined as the ambient-space vector $\nabla_{\mPhi}L$ that satisfies
% \begin{equation}
% \label{eq:grad_defn}
%     D_{\mPhi} L [\vxi] = g^{\mathbb{R}_{*}^{n\times r}}_{\mPhi}\left(\nabla_{\mPhi}L,\vxi\right) 
% \end{equation}
% where $D_{\mPhi} L [\vxi]$ denotes the directional derivative. 
Using the Lagrangian \eqref{eq:lagrangian}, it is straightforward to verify that the derivative of $L$ with respect to $\mPhi$ in the direction of $\vxi$ is given by
\begin{equation}
    D_{\mPhi} L [\vxi] = -2\sum_{i=0}^{N-1} \ve(t_i)^\intercal \mC(t_i)\left(\vxi \minv - \mPhi\minv \mPsi^\intercal\vxi \minv\right)\hat\vz(t_i),
\end{equation}
where we have used the identity
\begin{equation}
    D_{\mPhi}\left(\mPsi^\intercal\mPhi\right)^{-1}[\vxi] = -\left(\mPsi^\intercal\mPhi\right)^{-1} \mPsi^\intercal \vxi \left(\mPsi^\intercal\mPhi\right)^{-1}.
\end{equation}
Recalling the definition of the gradient in the ambient-space metric \eqref{eq:grad_Phi}, and using $\mPhi^\intercal \mPhi = \mI$ as commonly done in computation so that \eqref{eq:Phi_metric} reduces to the trace, we recover the gradient in \eqref{eq:grad_Phi}.
The other gradients in Proposition \ref{prop:gradient} can be obtained in a similar fashion and the proof is concluded. 

\bibliographystyle{siamplain}
\bibliography{refs}

\end{document}